\RequirePackage{fix-cm}
\documentclass[smallcondensed]{svjour3}     
\smartqed  
\usepackage{graphicx}
\usepackage{lscape}
\usepackage{amsmath}

\usepackage[export]{adjustbox}
\usepackage{algorithmicx}
\usepackage{algorithm}
\usepackage{algpseudocode}
\usepackage{graphicx}
\usepackage{booktabs}
\usepackage{subfig}
\usepackage{float}

\usepackage{xcolor}
\definecolor{reviewer1}{rgb}{0,0,0}
\definecolor{reviewer2}{rgb}{0,0,0}

\definecolor{Reviewer1}{rgb}{0,0,0}
\definecolor{Reviewer2}{rgb}{0,0,0}
\definecolor{myself}{rgb}{0,0,0}

\usepackage{bm}
\usepackage{amssymb}
\usepackage{booktabs}
\usepackage{natbib}
\usepackage{makecell}
\usepackage{hyperref}

\usepackage{array}\newcolumntype{m}[1]{>{\centering\arraybackslash}p{#1}}

 \journalname{}

\begin{document}

\title{A Simple Yet Efficient Method of Local False Discovery Rate Estimation Designed for Genome-Wide Association Data Analysis }

\titlerunning{A Simple Yet Efficient Method of Local False Discovery Rate Estimation}        

\author{Ali Karimnezhad}


\institute{A. Karimnezhad  \at Department of Mathematics and Statistics, University of Ottawa, Ottawa, Ontario, Canada \\
\email{a.karimnezhad@uottawa.ca}
       }
\date{Received: date / Accepted: date}

\maketitle

\begin{abstract}
In genome-wide association studies, hundreds of thousands of genetic features (genes, proteins, etc.) in a given case-control population are tested to verify existence of an association between each genetic marker and a specific disease. A popular approach in this regard is to estimate local false discovery rate (LFDR), the posterior probability that the null hypothesis is true, given an observed test statistic. However, the existing LFDR estimation methods in the literature are usually complicated. Assuming a chi-square model with one degree of freedom, which covers many situations in genome-wide association studies, we use the method of moments and introduce a simple, fast and efficient approach for LFDR estimation.  
We perform two different simulation strategies and compare the performance of the proposed approach with three popular LFDR estimation methods. We also examine the practical utility of the proposed method by analyzing a comprehensive 1000 genomes-based genome-wide association data containing approximately 9.4  million single nucleotide polymorphisms, and a microarray data set consisting of genetic expression levels for 6033 genes for prostate cancer patients. The R package implementing the proposed method is available on CRAN \url{https://cran.r-project.org/web/packages/LFDR.MME}.
\end{abstract}

\newpage

\section{Introduction}

 Genetic association studies deal with investigating an association between some disease traits and some genetic features, including genes, proteins, lipids and single nucleotide polymorphisms (SNPs). The investigation follows certain strategies to determine whether there exists some kind of statistical association.  In case-control studies, the investigation is started by the determination of  differences between the frequency of alleles or genotypes at genetic marker loci in individuals from a given population.  
Significant differences then reflect strong statistical evidence to claim for existence of association.

In genome-wide association studies (GWAS), $N$ SNPs (with $N$ being usually hundreds of thousands or millions) are genotyped in a given case-control population, \textcolor{reviewer1}{and are tested to verify whether each SNP $i$, $i=1,\ldots,N$, is associated with some disease. In this context, the null hypothesis $H_{0i}$ represents no association between SNP $i$ and the disease, and the alternative hypothesis  $H_{1i}$ represents an association in that population}. For an individual SNP $i$, the classic statistics deals with testing $H_{0i}$ versus $H_{1i}$ by verifying whether a test statistic $x_i$ falls inside some critical region $\mathcal{C}_\alpha$, where $\alpha$ is the significance level (type-I error or false positive rate). As an example, if $x_i$ represents an estimated allelic odds ratio (OR) for SNP $i$, then a critical region may be presented by $\mathcal{C}_\alpha=\{x_i:x_i<\chi_{1;1-\alpha/2}^2\,\text{or}\,x_i>\chi_{1;\alpha/2}^2\}$, where $\chi_{1;1-\gamma}^2$ denotes  $100(1-\gamma)\%$-quantile of the chi-square distribution with one degree of freedom. Alternatively, the single hypothesis problem might be tested by comparing the significance level $\alpha$ with the resulting p-value $p_i$, the smallest value of $\alpha$ such that $x_i\in \mathcal{C}_\alpha$. According to Fisher's scale of evidence for interpreting p-values, \textcolor{reviewer1}{the lower the p-value},  \textcolor{reviewer2}{the greater the evidence against the null hypothesis $H_{0i}$} \citep{Efron2012}. Either way, the procedure is simple and convenient to use, but \textcolor{reviewer2}{when $N$ is large,} the approach leads to a high-rate of false discoveries. 

 To overcome this challenging situation, several improvements on a set of given p-values have been introduced in the literature, see \cite{Sidak1968,Sidak1971}, \cite{Holm1979}, \cite{Simes1986} and \cite{Hochberg1988}. \textcolor{reviewer1}{Another method of p-value adjustment was introduced by \citet{Benjamini1995}. Their approach was designed to control the false discovery rate (FDR) and} led to many developments. For example, \citet{Storey2002}  \textcolor{reviewer2}{developed} a Bayesian approach for estimating FDR and  \citet{Efronetal2001} \textcolor{reviewer2}{outlined} an empirical Bayesian interpretation. The latter  \textcolor{reviewer2}{defined} FDR as the posterior probability that a null hypothesis $H_{0i}$ is true given that an observed test statistic $x_i$ falls within some critical region $\mathcal{C}$, i.e., $P(H_{0i}{\text{ is true}}|\mathcal{C})$. If the critical region $\mathcal{C}$ consists of only one point, FDR is referred to as the local FDR (LFDR), see \citet{Efron2012} and \citet{Marta2012}. 

\textcolor{reviewer1}{In the problem of testing $N$ SNPs simultaneously}, we assume that each SNP $i$, $i=1,\ldots,N$, is unassociated (with some disease) with prior probability $\pi_0$. We further suppose that the test statistic $x_i$ follows a null probability density function (pdf), say $f_0$, with \textcolor{reviewer2}{probability} ${\pi}_0$ and a non-null pdf, say $f_1$, with \textcolor{reviewer2}{probability} $\pi_1=1-\pi_0$. 
Then, LFDR for each SNP $i$ is defined as follows
\begin{equation}\label{LFDR}
\psi_i=P(H_{0i}{\text{ is true}}|x_i)=\frac{\pi_0f_0(x_i)}{f(x_i)},\end{equation}
where $f(x_i)=\pi_0 f_0(x_i) + \pi_1 f_1(x_i)$. In general, $f_0$ is assumed to be known (e.g., standard normal pdf, central chi-square pdf with some known degree(s) of freedom, etc.), $f_1$ is assumed to be a known pdf with some unknown parameter(s) (e.g., normal pdf with unknown mean and/or unknown variance,  chi-square pdf with some known degree(s) of freedom with some unknown non-centrality parameter, etc.), and $\pi_0$ is an unknown parameter. Such unknown parameters need to be estimated before making any inference. Estimated values are then replaced in equation (\ref{LFDR}) and the resulting estimated LFDR, say $\hat\psi_i$,  is compared to some pre-determined threshold. SNPs not passing the pre-specified threshold are deemed to be associated with the underlying disease.

Different strategies have been used for LFDR estimation in the literature. \citet{pan2003mixture} and \citet{Efron2004,Efron2007} perform the estimation task using a discrete mixture model, and \citet{Muralidharan2010}, \citet{Marta2012} and \citet{Yang2013}
use the maximum likelihood (ML) approach. \citet{Bickel2013} summarizes strengths and weakness of classic and empirical
Bayes estimation approaches.

\textcolor{Reviewer2}{In this paper, we assume} a multiple hypothesis testing problem in which $N$ SNPs in a given case-control population are tested. For each SNP $i$, $i=1,\ldots,N$, the null hypothesis $H_{0i}$ represents  that there is no association between SNP $i$ and a certain disease. We assume that $\pi_0$, the true proportion of unassociated SNPs, is unknown.  \textcolor{reviewer2}{Further, we suppose that $f_0$ represents a central chi-square pdf with one degree of freedom. 
We also assume that $f_1$ represents a chi-square pdf with one degree of freedom and an unknown non-centrality parameter $\lambda$}. We will discuss in the forthcoming section that these assumptions are not restrictive and are made in many genetic association studies.

\textcolor{reviewer1}{\citet{Efron2004} introduces a histogram-based (HB) method of LFDR estimation, in which $f_0$ in equation (\ref{LFDR}) is assumed to follow a standard (or empirical) normal distribution and the mixture density $f$ is estimated by a Poisson regression model. The HB approach is fast and does not require independency of test statistics. However, it has been designed to perform well when $\pi_0$ is higher than 0.9 only.} Being concentrated on the chi-square model, \citet{Marta2012} as well as \citet{karimnezhad2020incorporating}  assume that the test statistics $x_i$, $i=1,\ldots,N$, are independent and estimate the corresponding LFDR by using the ML approach, i.e., 
\begin{equation}\label{Estimated.LFDR}
\hat{\psi}_i=\frac{\hat{\pi}_0f_0(x_i)}{\hat{\pi}_0f_0(x_i)+(1-\hat{\pi}_0)f_{\hat{\lambda}}(x_i)},\end{equation}
where 
\begin{equation}(\hat{\pi}_0,\hat{\lambda})=\arg\max_ {\pi_0\in[0,1],\lambda\in[c,d]}\prod_{i=1}^N \left(\pi_0f_0(x_i)+(1-\pi_0)f_\lambda(x_i)\right),\label{ML_arg}
\end{equation}
$f_0$ represents the chi-square pdf with one degree of freedom, $f_\lambda$ represents the chi-square pdf with one degree of freedom and non-centrality parameter $\lambda$, and $c$ and $d$ are known positive bounds. Although this approach leads to somehow sensible estimators, the independency assumption might be unrealistic in genetic association studies due to linkage disequilibrium (LD), \textcolor{reviewer1}{which refers to a nonrandom association of alleles at two or more loci. For more information regarding the concept of LD and its measurement, one may refer to \cite{Slatkin2008} and \cite{Zheng2012} among many others}. The resulting LFDR estimates highly depend on the bounds of $c$ and $d$, and inappropriate choices for these bounds can negatively affect the estimation precision. As another downside, this estimation procedure is time-consuming, and the processing time increases with the number of SNPs as well as the length on the interval $[c,d]$. Alternative to these approaches, we provide a simple yet efficient algorithm that estimates LFDRs without assuming independency. We show that the resulting estimator has a high precision as long as $N$, the number of SNPs to be tested, is large. \textcolor{reviewer1}{Not only the proposed approach is fast, but also it provides an explicit} form of estimators of the proportion rate $\pi_0$ and the non-centrality parameter $\lambda$. As a result, unlike many algorithms including the ML approach of  \citet{Marta2012} and  the HB approach of \citet{Efron2004,Efron2012}, it provides an explicit form of the corresponding LFDR estimator.

The structure of this paper is as follows. \textcolor{reviewer1}{In Section \ref{sec:2}, we briefly review some basic terms and concepts in genetics and genetic epidemiology.} In Section \ref{sec:3}, we quickly review frequent measures for genetic association studies used in the literature. We discuss that many genetic association studies reduce to a multiple hypothesis testing problem in a chi-square model.  In Section \ref{sec:4}, we present our proposed empirical Bayes approach.  Section \ref{sec:5} is devoted to evaluating performance of the proposed approach. We follow two different simulation strategies and use the mean squared error (MSE) as a common measure of performance. In Section \ref{sec:6}, we apply the proposed approach and analyze two different sets of real data, including a set of coronary artery disease data and another set of microarray data.  \textcolor{reviewer2}{Finally, we provide} some discussion and concluding remarks in Section \ref{sec:7}.

\section{Basic Notions}\label{sec:2}
\textcolor{reviewer1}{In this section, we briefly review some basic terms and definitions in genetics and genetic epidemiology.}

\textcolor{reviewer1}{Human beings have normally 23 pairs of chromosomes in each cell nucleus. Each chromosome consists of a molecule of deoxyribonucleic acid (DNA). DNA is made up of a long sequence of four different nucleotides, namely adenine ($A$), thymine ($T$), guanine ($G$) and cytosine ($C$). A gene is a series of DNA sequences that contain genetic information. A locus is the location of a gene or any DNA sequence on a chromosome pair. When the location of a gene or DNA sequence on a chromosome is known, and that sequence varies in the population, it is also called a genetic marker. A genetic marker can be used to identify individuals and can be described as a variation that can be observed. An allele is an alternative DNA sequence that can occur at a particular locus. Since chromosomes are present in pairs, a person's gene or marker at a given locus has two alleles, one on each chromosome. In the population, however, a gene or marker could have more than two alleles. The combination of alleles that an individual possesses for a specific gene is referred to as their genotype. In this paper, we focus on diallelic markers, which have only two alleles in the population. A SNP is a commonly used diallelic marker that varies in individuals owing to the difference of a single nucleotide ($A$, $T$, $C$, or $G$) in the DNA sequence. }

We consider a diallelic marker locus with a typical allele (wild-type) $A$ and the alternative (risk) allele $B$, and let $p$ represent the frequency of the risk allele, i.e., $p=P(B)$ and $1-p=P(A)$.  We also denote the corresponding genotypes by $G_0=AA$, $G_1=AB$ and $G_2=BB$ (we do not distinguish between $AB$ and $BA$). Then, genotype frequencies in the population are given by  $g_j=P(G_j)$, $j=0,1,2$, where $g_0=(1-p)^2+p(1-p)F$, $g_1=2p(1-p)(1-F)$ and $g_2=p^2+p(1-p)F$, in which $F$ is Wright's coefficient of inbreeding. For humans, $F$ is usually taken to be between 0 and 0.05. Under Hardy-Weinberg (HW) equilibrium, $F=0$. Thus, when HW proportions hold in the population, $g_0=(1-p)^2$, $g_1=2p(1-p)$ and $g_2=p^2$. Now, let the prevalence of the disease be $k=P(\text{case})$, and for  $j=0,1,2$, define $ v_j = P(\text{case} |G_j ) $, the probability of having a disease given a specific genotype at the marker for genotype $G_j$. Obviously, $k=\sum_{j=0}^2 v_jg_j$. Depending on a chosen genetic model, the penetrances $v_j$ have certain relationships with themselves. If the genetic model is additive, then $v_1=\frac{v_0+v_2}{2}$. For recessive, multiplicative and dominant models, $v_1=v_0$, $v_1=\sqrt{v_0 v_2}$ and $v_1=v_2$, respectively, see  \citet{Zheng2012} for more details.

\section{Common Measures for Association}\label{sec:3}
Genetic association is usually measured for each individual SNP separately. The data for each SNP can be summarized in a contingency table of either genotype counts or allele counts by disease status (case or control). Table \ref{Tab1} represents genotype counts at marker $M$ based on a sample of $r$ cases and $s$ controls. Intuitively, the \textcolor{reviewer2}{probability} of observing genotype $G_j$, $j=0,1,2$, provided that an individual is known to belong to the case group is estimated by $\frac{r_j}{r}$, which is in fact the ML estimate of  $p_j = P(G_j | \text{case})$. Similarly, $q_j = P(G_j | \text{control})$ is the true \textcolor{reviewer2}{probability} of having genotype $G_j$ for a control individual which, with the notations in Table 1, is estimated by \textcolor{reviewer2}{$\frac{s_j}{s}$}. It is easy to verify using the Bayes principle that $p_j=\frac{v_jg_j}{k}$ and $q_j=\frac{(1-v_j)g_j}{1-k}$. Table \ref{Tab2} represents allele counts  at marker $M$ based on a sample of $r$ cases and $s$ controls.

\begin{center}
\tabcolsep=0.5cm
\begin{table}[h]
\caption{A typical 2$\times$3 table with \textcolor{reviewer2}{$r$ cases and $s$ controls.}
}
\label{Tab1}
\begin{footnotesize}
\begin{center}
\begin{tabular}{c|ccc|c}
\hline 
 & $AA$ & $AB$ & $BB$ & $Total$\tabularnewline
\hline

Case & $r_{0}$ & $r_{1}$ & $r_{2}$ & $r$\tabularnewline

Control & $s_{0}$ & $s_{1}$ & $s_{2}$ & $s$\tabularnewline
\hline 
Total & $n_{0}$ & $n_{1}$ & $n_{2}$ & $n$\tabularnewline

\hline
\end{tabular}
\end{center}
\end{footnotesize}
\end{table}
\end{center}

\begin{center}
\tabcolsep=0.5cm
\begin{table}[h]
\caption{Typical allele counts of case-control samples for a single marker.}
\label{Tab2}
\begin{footnotesize}
\begin{center}
\begin{tabular}{c|cc|c}
\hline 
 & $A$ & $B$ & $Total$\tabularnewline
\hline
Case & $2r_{0}+r_1$ & $2r_{2}+r_1$ & $\textcolor{reviewer2}{2r}$\tabularnewline

Control & $2s_{0}+s_{1}$ & $2s_{2}+s_1$ & $\textcolor{reviewer2}{2s}$\tabularnewline
\hline 
Total & $2n_{0}+n_1$ & $2n_{2}+n_1$  & $\textcolor{reviewer2}{2n}$\tabularnewline
\hline
\end{tabular}
\end{center}
\end{footnotesize}
\end{table}
\end{center}

Contingency tables play a key role in summarizing genetic association data. Under the null hypothesis that there is no association between genotypes in Table \ref{Tab1} (or alleles  in Table \ref{Tab2}) 
and the disease, the same relative genotype (or allele) frequencies in both case and control groups are expected. 
 Perhaps Pearson's test is the most convenient association test in contingency tables. Pearson's test statistic for the genotypes in Table \ref{Tab1} can be presented by
\begin{equation*}\label{Pearson}
 T_{P}=\sum_{j=0}^2\frac{(r_j-n_jr/n)^2}{n_jr/n}+\sum_{j=0}^2\frac{(s_j-n_js/n)^2}{n_js/n},
\end{equation*}
which approximately follows a chi-square distribution with two degrees of freedom. In the same way, the test can be applied to allele counts in Table \ref{Tab2}, and the resulting test statistic approximately follows a chi-square distribution with one degree of freedom, see \citet{Clarke2011} and \citet{Zheng2012} among many others.
In fact, any test statistic in genetic association studies by means of a contingency table follows a chi-square distribution with at most two degrees of freedom, see Table 2 of \citet{Clarke2011}. Readers may also refer to \citet{Clarke2011} for a summary of strengths and weaknesses of allelic and genotypic tests of association, as well as the differences between the different models of penetrance.  \textcolor{reviewer2}{There are other contingency table-based association tests in the literature. For example, Cochran-Armitage trend test is used in situations where some kind of trend in risk of developing the disease with increasing number of the risk allele in three genotypes is determined. See  \citet{Clarke2011} and \citet{Zheng2012} for more details}.

\textcolor{reviewer2}{Odds ratio (OR) is another common measure of association between genotypes and diseases. It compares the odds of disease in an individual carrying one genotype to the odds of disease in an individual carrying a different genotype. For a diallelic marker, two genotypic ORs can be defined as given below
\begin{equation}
OR_{j}=\frac{v_j(1-v_0)}{v_0(1-v_j)},\quad j=1,2.\label{eq:OR-1}
\end{equation}
$OR_1$ compares the odds of disease between individuals carrying genotype $AB$ and those carrying $AA$, and $OR_2$ compares the odds of disease between individuals carrying genotype $BB$ and those carrying $AA$. In case of no association, $OR_1=OR_2=1$. The corresponding test statistic, say $T_{OR_{j}}$, follows a standard normal distribution. If so, $T_{OR_{j}}^2\sim\chi_1^2$.  An OR can also be defined for the allele counts in Table \ref{Tab2}, and similarly, the corresponding test statistic follows a chi-square distribution with one degree of freedom.}

A more flexible analysis for GWAS is based on the logistic regression model. Let the binary random variable $Y_{i}$ represent whether \emph{i}th individual belongs to the case group
($Y_{i}=1$) or the control group ($Y_{i}=0$), and let $X_{i}$ denote the genotype of individual $i$ for an arbitrary SNP so that $X_{i}(G_0)=0$, $X_{i}(G_1)=1$ and $X_{i}(G_2)=2$. Then, the logistic model is defined as $\ln \frac{\theta_{i}}{1-\theta_{i}}=\beta_{0}+\beta_{1} X_{i}$, where $\theta_{i}=E[Y_{i}|X_{i}]$ is the expected value of phenotype given a genotype for an arbitrary SNP. With this setting, the multiple hypothesis testing problem reduces to testing $H_{0i}:\beta_{1}=0$ versus $H_{1i}:\beta_{1}\neq 0$, $i=1,\ldots,N$. The corresponding test statistic is computed by $T_{L}=\frac{\hat{\beta}_{1}}{\sqrt{\widehat{Var}(\widehat{\beta}_{1})}}$.  Under the null hypothesis of no association, it follows that 
$T_{L}\sim N(0,1)$ or equivalently  $T_{L}^2\sim \chi_1^2$. For more details see \citet{
Marta2012}, \citet{Yang2013} and \citet{karimnezhad2020incorporating}. 

\section{A novel empirical Bayesian method}\label{sec:4}
As discussed in the preceding section, many genetic association studies reduce to a multiple hypothesis testing problem in which the corresponding test statistics follow a chi-square distribution. \textcolor{reviewer2}{To perform this hypothesis test}, we propose to apply and  estimate LFDRs by using a simple and efficient empirical Bayes approach, as we present below.

Suppose that $N$ SNPs have been genotyped in a case-control population, and that the goal is to test the null hypothesis $H_{0i}$, $i=1,\ldots,N$, indicating that there is no association between SNP $i$ and the disease, \textcolor{reviewer1}{versus} its alternative hypothesis  $H_{1i}$. Suppose that for each SNP $i$, the test statistic $x_i$ has already been computed using any of the approaches reviewed in the preceding section. \textcolor{reviewer2}{Then, define an indicator variable $\mu_{i}$ such that $\mu_{i}=0$ when $H_{0i}$ is true and $\mu_{i}=1$ when $H_{1i}$ is true. Also, let $\pi_0\in[0,1]$ be the true proportion of SNPs not associated with the disease. If so, $\pi_0=P(\mu_{i}=0)$. In fact, the indicator variable $\mu_i$ assigns the \textcolor{reviewer2}{probability} $\pi_0$ to each null hypothesis to be true. }
Define $\theta$ to be a two-state variable so that it takes 0 if $\mu_i=0$, and takes a positive value $\lambda$ if $\mu_i=1$. 
Further, \textcolor{reviewer2}{assume that the test statistic $X_i$ follows $\chi_{1,\theta}^2$, the chi-square distribution with one degree of freedom and non-centrality parameter $\theta$}.  
This model can be expressed by the following hierarchical model
\[
\begin{cases}
\textcolor{reviewer2}{X_{i}|\theta\sim \chi_{1,\theta}^2},\quad i=1,\ldots,N,\\
\theta|\mu_{i}\sim {\mu_i}\delta_{\mu_i}+\lambda\delta_{1-\mu_{i}},\\
\mu_{i}\sim Bernoulli(1-\pi_{0}),
\end{cases}
\]
where \begin{equation*}
\delta_a=
\begin{cases}
0 \quad \text{if }a\neq0,\\
1 \quad \text{if } a=0.
\end{cases}
\end{equation*}
It is interesting to mention that a similar hierarchical Bayes model has already been applied in detecting variants in the analysis of next generation sequencing data, see \cite{Zhao2013}.

 \textcolor{reviewer2}{Focusing on the chi-square distribution with one degree of freedom, we obtain the following simplified LFDR estimation.}
\textcolor{reviewer2}{\begin{theorem}\label{thm1} Let $X_{i}$, $i=1,\ldots,N$, follow the chi-square distribution with one degree of freedom and the non-centrality parameter $\lambda$. \begin{itemize}
\item [{(i)}] The pdf of an observation $x_i$ can be expressed by
 \textcolor{reviewer2}{\begin{equation}\label{f_lambda}
f_{\lambda}(x_{i})=\frac{1}{2\sqrt{2\pi x_{i}}}\left(e^{-\frac{1}{2}(\sqrt{x_{i}}-\sqrt\lambda)^{2}}+e^{-\frac{1}{2}(\sqrt{x_{i}}+\sqrt\lambda)^{2}}\right).
\end{equation}}
\item [{(ii)}] Let $x_{i}$ be an observation from the mixture pdf $f(x_{i})=\pi_{0}f_{0}(x_{i})+(1-\pi_{0})f_{\lambda}(x_{i})$.
Then, the LFDR based on observing $x_{i}$ is given by 
\begin{equation}\label{simplified.LFDR}
\psi(x_{i})=\pi_{0}\left(\pi_{0}+(1-\pi_{0})e^{-\frac{1}{2}\lambda}\cosh(\sqrt{\lambda x_{i}})\right)^{-1}.
\end{equation}
\item [{(iii)}] $\psi(x_{i})$ is a strictly decreasing function of $x_i$, and 
\[\lim_{x_i\rightarrow 0}\psi(x_{i})=\frac{\pi_{0}}{\pi_{0}+(1-\pi_{0})e^{-\frac{1}{2}\lambda}},\quad
\lim_{x_i\rightarrow \infty}\psi(x_{i})=0.\]
\item [{(iv)}] For a given threshold $u\in[0,1]$, define $k_{u}(\pi_{0},\lambda)= \frac{\pi_0}{1-\pi_0}  \frac{1-u}{u} e^{\frac{\lambda}{2}}$. Then,  $\psi(x_{i})< u$ if and only if $x_{i}> h_{u}(\pi_{0},\lambda)$,
where 
\begin{equation}\label{hu}
h_{u}(\pi_{0},\lambda)=
\begin{cases}
0 & \mathrm{if \,} k_{u}(\pi_{0},\lambda) \leq 1,\\
\frac{1}{\lambda}\left[\ln\left(k_{u}(\pi_{0},\lambda)+\sqrt{k_{u}^{2}(\pi_{0},\lambda)-1}\right)\right]^{2} &  \mathrm{if \,} k_{u}(\pi_{0},\lambda) > 1.
\end{cases}
\end{equation}
 \end{itemize}
\end{theorem}}
\begin{proof} \textcolor{reviewer2}{(i) Suppose that a random variable $Y_i$ follows a normal distribution with mean $\sqrt \lambda$ and unity variance. Then, it is well-known that $X_i=Y_i^2$ follows a non-central chi-square distribution with one degree of freedom and non-centrality parameter $\lambda$, see \cite{Shao2007} among many others. Let $F_{X_i}(\cdot|\lambda)$ and $F_{Y_i}(\cdot|\lambda)$ denote the cumulative distribution function of $X_i$ and $Y_i$ given $\lambda$, respectively. Then, it is easy to verify that 
\[
F_{X_i}(x_{i}|\lambda)=F_{Y_i}(\sqrt{x_{i}}|\lambda)-F_{Y_i}(-\sqrt{x_{i}}|\lambda).
\]
Now differentiating both sides w.r.t. $x_i$ results in (\ref{f_lambda}).\\(ii) The LFDR in (\ref{simplified.LFDR}) is derived using equation (\ref{LFDR}) with $f(x_{i})=\pi_{0}f_{0}(x_{i})+(1-\pi_{0})f_{\lambda}(x_{i})$ in which $f_{0}(x_{i})$ and $f_{\lambda}(x_{i})$ are substituted from part (i).\\
(iii) For a fixed $\lambda>0$, when $x_i$ increases  from 0 to $\infty$,  $\cosh(\sqrt{\lambda x_{i}})$ strictly increases from 1 to $\infty$. Thus, $\psi(x_i)$ strictly decreases from $\pi_{0}\left(\pi_{0}+(1-\pi_{0})e^{-\frac{1}{2}\lambda}\right)^{-1}$ to 0. \\
(iv) From (\ref{simplified.LFDR}) notice that
\begin{align}\label{psi_cosh}
\psi(x_{i})< u \iff & \frac{\pi_0}{1-\pi_0}  \frac{1-u}{u} e^{\frac{\lambda}{2}}< \cosh(\sqrt{\lambda x_{i}})\nonumber\\
\iff & k_{u}(\pi_{0},\lambda)< \frac{w_i+w_i^{-1}}{2}
\end{align}
where $k_{u}(\pi_{0},\lambda)=\frac{\pi_0}{1-\pi_0}  \frac{1-u}{u} e^{\frac{\lambda}{2}}$ and $w_i=e^{\sqrt{\lambda x_{i}}}$ for simplicity. Note that $k_{u}(\pi_{0},\lambda)\in[0,\infty)$, and since $x_i$ only takes on positive values, $w_i$ has to fall in the interval $(1,\infty)$. By multiplying both sides of the last inequality in (\ref{psi_cosh}) by $w_i$, we observe that $\psi(x_{i})< u$ if and only if  $g(w_i)=w_i^2-2k_{u}(\pi_{0},\lambda)w_i+1$, as a quadratic function of $w_i$, is positive. Now, we need to find values of $w_i$ for which $g(w_i)>0$. It is easy to verify that $g(w_i)$ has a unique minimum at $w_i=k_{u}(\pi_{0},\lambda)$ but the number of possible roots needs to be verified according to the sign of the discriminant of $g(w_i)$, i.e., $\Delta=4(k_{u}^2(\pi_{0},\lambda)-1)$. Depending on the value of $k_{u}(\pi_{0},\lambda)$, one of the following three cases may happen. \\
Case 1. If  $0\leq k_{u}(\pi_{0},\lambda)<1$, then $\Delta<0$ and  $g(w_i)$ has no roots. The fact that the quadratic function $g(w_i)$ has a unique minimum at $w_i=k_{u}(\pi_{0},\lambda)$ yields that for every $w_i\in(1,\infty)$, $g(w_i)> g(k_{u}(\pi_{0},\lambda))=1-k_{u}^2(\pi_{0},\lambda)>0$. Thus, for every $x_i>0$, $\psi(x_{i})< u$. \\
Case 2. If $k_{u}(\pi_{0},\lambda)=1$, then $\Delta=0$. This implies that $g(w_i)$ may have a unique root at $w_i=k_{u}(\pi_{0},\lambda)=1$ but this does not occur as $w_i\notin(1,\infty)$. Thus, for every value of $w_i\in (1,\infty)$, $g(w_i)> g(1)=0$, or equivalently, for every $x_i>0$, $\psi(x_{i})< u$. \\
Case 3. If   $k_{u}(\pi_{0},\lambda)>1$, then $\Delta>0$ and  $g(w_i)$ has two roots at $w_{i,1}=k_{u}(\pi_{0},\lambda)-\sqrt{k_{u}^2(\pi_{0},\lambda)-1}$ and $w_{i,2}=k_{u}(\pi_{0},\lambda)+\sqrt{k_{u}^2(\pi_{0},\lambda)-1}$. But $w_{i,1}$ is an unacceptable root, as it is negative. Thus, the only eligible root is $w_i=w_{i,2}$. Since for every $w_i> k_{u}(\pi_{0},\lambda)$, $g(w_i)$ is strictly increasing, it is concluded that for every $w_i> w_{i,2}> k_{u}(\pi_{0},\lambda)$, $g(w_i)> g(w_{i,2})=0$, or equivalently, for every $x_i>\frac{1}{\lambda}\left[\ln\left(k_{u}(\pi_{0},\lambda)+\sqrt{k_{u}^2(\pi_{0},\lambda)-1}\right)\right]^2$, $\psi(x_{i})< u$. }
\qed\end{proof}

Equation (\ref{simplified.LFDR}) is in fact a simplified version of equation (\ref{LFDR}) which is applicable to genetic association studies. To estimate the LFDR, the parameters $\pi_0$ and $\lambda$ in (\ref{simplified.LFDR}) need to be estimated. In this regard, we propose the method of moments (MM) estimation, which suggests that unknown parameters in a model should be estimated by matching theoretical moments with the appropriate sample moments \citep{Matyas1999}. 

\begin{theorem}\label{thm2} With the setting of Theorem \ref{thm1}, let $m_{1}=\frac{1}{N}\sum_{i}X_{i}$ and $m_{2}=\frac{1}{N}\sum_{i}X_{i}^{2}$ represent the first and the second moments, respectively. Then, MM estimators of $\lambda$ and $\pi_0$ are respectively given by 
\begin{equation}\label{p0.hat}\hat{\lambda}=\frac{m_{2}-3}{m_{1}-1}-6,\quad 
\widehat{\pi}_{0}=1-\frac{m_{1}-1}{\hat{\lambda}}.
\end{equation}
\end{theorem}
\begin{proof}
By using the properties of conditional expectation, observe that 
\begin{align*}
E[X_{i}]  &=E_{\theta}\left[E[X_{i}|\theta]\right]\\
 & =1+E_{\mu_{i}}\left[E[\theta|\mu_{i}]\right]\\
  &=1+(1-\pi_{0})E_{\theta}[\theta|\mu_{i}=1],
\end{align*}
and
\begin{align*}
E[X_{i}^{2}]  &=E_{\theta}\left[E[X_{i}^{2}|\theta]\right]\\
 & =E_{\mu_{i}}\left[E_{\theta}[(\theta+1)^{2}+2(1+2\theta)|\mu_{i}]\right]\\
 &=6m_{1}-3+(1-\pi_{0})E_{\theta}[\theta^{2}|\mu_{i}=1].
\end{align*}
Note that $E_{\theta}[\theta|\mu_{i}=1]=E_{\theta}[\theta^{2}|\mu_{i}=1]=\lambda$. Now, equating the above expectations with the first and second moments leads to (\ref{p0.hat}).
\qed
\end{proof}

Consistency of the MM estimators  \citep{Matyas1999} guarantees that  $\widehat\lambda$ and  $\widehat\pi_0$  converge in probability to  $\lambda$ and  $\pi_0$, respectively. We show in the next section that $\widehat\lambda$ and  $\widehat\pi_0$ estimate the true parameters  $\lambda$ and  $\pi_0$ very well. 

To make an inference regarding association between $i$th SNP and the disease, one may compute the estimated LFDR ${\widehat\psi}_i$ by replacing estimates of $\lambda$ and  $\pi_0$  from equation (\ref{p0.hat}) into equation (\ref{simplified.LFDR}). Therefore, if for a given threshold $u$, $\widehat\psi _i<u$, the null hypothesis $H_{0i}$ is rejected. Otherwise, there is no evidence to conclude an association. An alternative approach would be to replace estimates of $\lambda$ and  $\pi_0$  from equation (\ref{p0.hat}) into $h_u(\pi_0,\lambda)$ in equation (\ref{hu}). Then, $H_{0i}$ is rejected only if the test statistic $x_i$ is greater than $h_u(\widehat\pi_0,\widehat\lambda)$. The second approach is simpler and more convenient, and unlike the existing methods in the literature, it allows for performing multiple hypothesis testing by just comparing each of the test statistics $x_i$ with a purely data-based \textcolor{reviewer2}{threshold,} i.e., $h_u(\widehat\pi_0,\widehat\lambda)$.

The threshold $u$ in (\ref{hu}) can be chosen according to a subjective belief. A conventional choice would be to choose $u=0.2$ to identifying ``interesting cases",  see \citet{Efron2012}.  It can also be chosen based on an objective belief. In this regard, \citet{karimnezhad2020incorporating}  follow a decision theoretic approach in which for a binary decision rule $\delta_{i}$, the null hypothesis $H_{0i}$
is rejected if $\delta_{i}=1$, and is not rejected if $\delta_{i}=0$. They use the following loss function
\[
L(\mu_{i},\delta_{i})=\begin{cases}
0 & \delta_{i}=\mu_{i}=1\,\text{or}\,\delta_{i}=\mu_{i}=0,\\
l_{I} & \delta_{i}=1,\,\mu_{i}=0,\\
l_{II} & \delta_{i}=0,\,\mu_{i}=1,
\end{cases}
\]
where $l_{I}$ and $l_{II}$ are loss values incurred due to making
type I and type II errors, respectively. The resulting Bayes estimator of the parameter $\mu_i$ is then given by
\begin{equation}
\delta_{i}=\begin{cases}
1 & \mathrm{if\:}\widehat{\psi}_{i}<\frac{l_{II}}{l_{I}+l_{II}},\\
0 & \mathrm{if\:}\widehat{\psi}_{i}\geq\frac{l_{II}}{l_{I}+l_{II}}.
\end{cases}\label{eq:Bayes rule 0-1 loss}
\end{equation}
Now, it can be verified using Theorem \ref{thm1} that the Bayes rule $\delta_i$ in equation (\ref{eq:Bayes rule 0-1 loss}) reduces to the following Bayes rule
\begin{equation}
\delta_{i}^u=\begin{cases}
1 & \mathrm{if\:}x_{i}> h_{u}(\widehat{\pi}_{0},\widehat{\lambda}),\\
0 & \mathrm{if\:}x_{i} \leq h_{u}(\widehat{\pi}_{0},\widehat{\lambda}),
\end{cases}\label{eq:NewBayes rule 0-1 loss-1}
\end{equation}
with $u=\frac{l_{II}}{l_{I}+l_{II}}$. This Bayes rule is simpler and more convenient than the Bayes rule $\delta_i$ in equation (\ref{eq:Bayes rule 0-1 loss}), due to the fact that it is based on the observed test statistic $x_i$ and estimates of $\pi_0$ and $\lambda$, which are available through equation (\ref{p0.hat}). In fact, equation (\ref{eq:NewBayes rule 0-1 loss-1}) illustrates that, unlike many existing algorithms in the literature, one may perform a multiple hypothesis testing comparison by just comparing their observed test statistics $x_i$ and the data-based function $h_{u}(\widehat{\pi}_{0},\widehat{\lambda})$.

\section{Simulation}\label{sec:5}
To illustrate the performance of the proposed LFDR estimation approach, we conduct simulations using two different strategies. \textcolor{reviewer1}{In addition to our proposed method, we investigate the performance of the HB method of \citet{Efron2004}, the ML method of  \citet{Marta2012}, and the FDR correction method of \citet{Benjamini1995}. We will briefly refer to the latter one by BH.}

\subsection{First simulation study}\label{subsec:5.1}
We follow the simulation strategy used in \citet{karimnezhad2020incorporating}. We take advantage of the fact that squared of log transformation of OR follows a chi-square distribution with one degree of freedom,  and that as reviewed in Section \ref{sec:3}, many algorithms in genetic association studies reduce to a chi-square model with one degree of freedom. For each iteration in our simulation study, we assume there are a total number of $N$ SNPs to be tested, of which $N_0$ SNPs are associated.  \textcolor{reviewer2}{The proportion of unassociated SNPs is then given by $\pi_0=1-\frac{N_0}{N}$}. We generate $z_i$ from $N(\log(OR),\sigma^2)$-distribution, where $\sigma^2$ is known, for $i=1,\ldots,N_0$, $OR\neq1$, and for $i=N_0+1,\ldots,N$, $OR=1$. Obviously, $x_i=(\frac{z_i}{\sigma})^2 \sim\chi_{1,\lambda}^2$, where for $i=1,\ldots,N_0$, $\lambda=\left(\frac{\log(OR)}{\sigma}\right)^2$, and for $i=N_0+1,\ldots,N$, $\lambda=0$.

\begin{algorithm}[h]
\caption{First simulation strategy.}\label{alg1}
 \begin{enumerate}
\item[] Step 1. Specify $N$, $N_0$, $OR$ and $\sigma^2$. By these values, $\pi_0=1-\frac{N_0}{N}$ and $\lambda=\left(\frac{\log(OR)}{\sigma}\right)^2$.
\item[] Step 2. Take $j=1$.
\item[] Step 3. Generate $z_{1},\ldots,z_{N_0}$ from $N(\log(OR),\sigma^{2})$-distribution.
\item[] Step 4. Generate $z_{N_0+1},\ldots,z_{N}$ from $N(0,\sigma^{2})$-distribution.
\item[] Step 5. Compute $x_{i}=(\frac{z_{i}}{\sigma})^{2}$, $i=1,\ldots,N$, the chi-square test statistics.
\item[] Step 6. Estimate $\psi_{i}$, $i=1,\ldots,N$, using the MM, ML and HB approaches. Denote the corresponding estimate by $\widehat{\psi}_{i}^M$, $M\in\{MM,ML,HB\}$. Also, denote the corresponding estimates of $\pi_0$ and $\lambda$ by $\widehat{\pi}_0^M$ and  $\widehat{\lambda}^M$. 
\item[] Step 7. Compute errors in estimating $\pi_0$, $\lambda$ and $\psi_i$ by $$E_{\pi_0}^{M,j}=(\widehat{\pi}_0^M-\pi_0)^2,\,E_{\lambda}^{M,j}=(\widehat{\lambda}^M-\lambda)^2,\,E_{\psi}^{M,j}=\frac{1}{N}\sum_{i=1}^{N}(\widehat{\psi}_i^M-\psi_i)^2.$$
\item[]  Step 8. Increase $j$ by one and repeat Steps 3 to 7  \textcolor{reviewer2}{until $j=b$ times. Then, in correspondence with each method $M\in\{MM,ML,HB\}$, compute
\begin{eqnarray*}
MSE_{\pi_0}^M=\frac{1}{b}\sum_{j=1}^{b}E_{\pi_0}^{M,j}, \quad  MSE_{\lambda}^M=\frac{1}{b}\sum_{j=1}^{b}E_{\lambda}^{M,j},\quad MSE_{\psi}^M=\frac{1}{b}\sum_{j=1}^{b}E_{\psi}^{M,j}.
\end{eqnarray*}} 
\end{enumerate}
\end{algorithm}

We took the steps in Algorithm \ref{alg1} with $N=1,000,000$, \textcolor{Reviewer2}{$\pi_0=0, 0.05$, $0.10(0.1)$ $0.90, 0.95, 1$} and $b=100$. We also took $OR=1.5$ and $\sigma^2=0.01$. By these choices, the true non-centrality parameter is 16.44. \textcolor{Reviewer1}{Figure \ref{fig:mean_sd_pi0}(a)-(b) represents plots of mean and standard deviation (sd) of $\widehat{\pi}_0$ at different selected true $\pi_0$ values. Also, Figure \ref{fig:mean_sd_lambda}(a)-(b) displays plots of mean and standard deviation of $\widehat{\lambda}$ at the same $\pi_0$ values.} Note that the HB approach outputs estimates of ${\psi}_{i}$ and $\pi_0$, while the MM and ML approaches output estimates of ${\psi}_{i}$, $\pi_0$ and $\lambda$.  \textcolor{Reviewer1}{From Figure  \ref{fig:mean_sd_pi0}(a), we observe that mean of $\hat{\pi}_0$ values computed by the HB approach is almost identical to the true $\pi_0$ values only when $\pi_0$ is close to 1. This is not surprising, as according to \citet{Efron2004}, the HB approach has been designed to perform well when $\pi_0$ is higher than 0.9. Figure \ref{fig:mean_sd_pi0}(b) also reveals that the corresponding variance values are higher than the ones computed by the MM and ML approaches (except at $\pi_0=1$). The performance of the HB approach was weak when $\pi_0$ is less than 0.7 (it failed when $\pi_0<0.5$). As reflected in Figure \ref{fig:mean_sd_pi0}(a)-(b), both the ML and MM approaches estimated $\pi_0$ very well (except at $\pi_0=1$). From Figure \ref{fig:mean_sd_lambda}(a), we observe that both the ML and MM estimated the non-central parameter $\lambda$ very well when $\pi_0\leq 0.9$, and their performance weakened as $\pi_0$ increases. We also observe from Figure \ref{fig:mean_sd_lambda}(b) that the ML approach led to lower variance values than the MM approach, especially when $\pi_0=1$.}
 
\begin{figure}
\centering
\subfloat[]{
  \includegraphics[width=60mm]{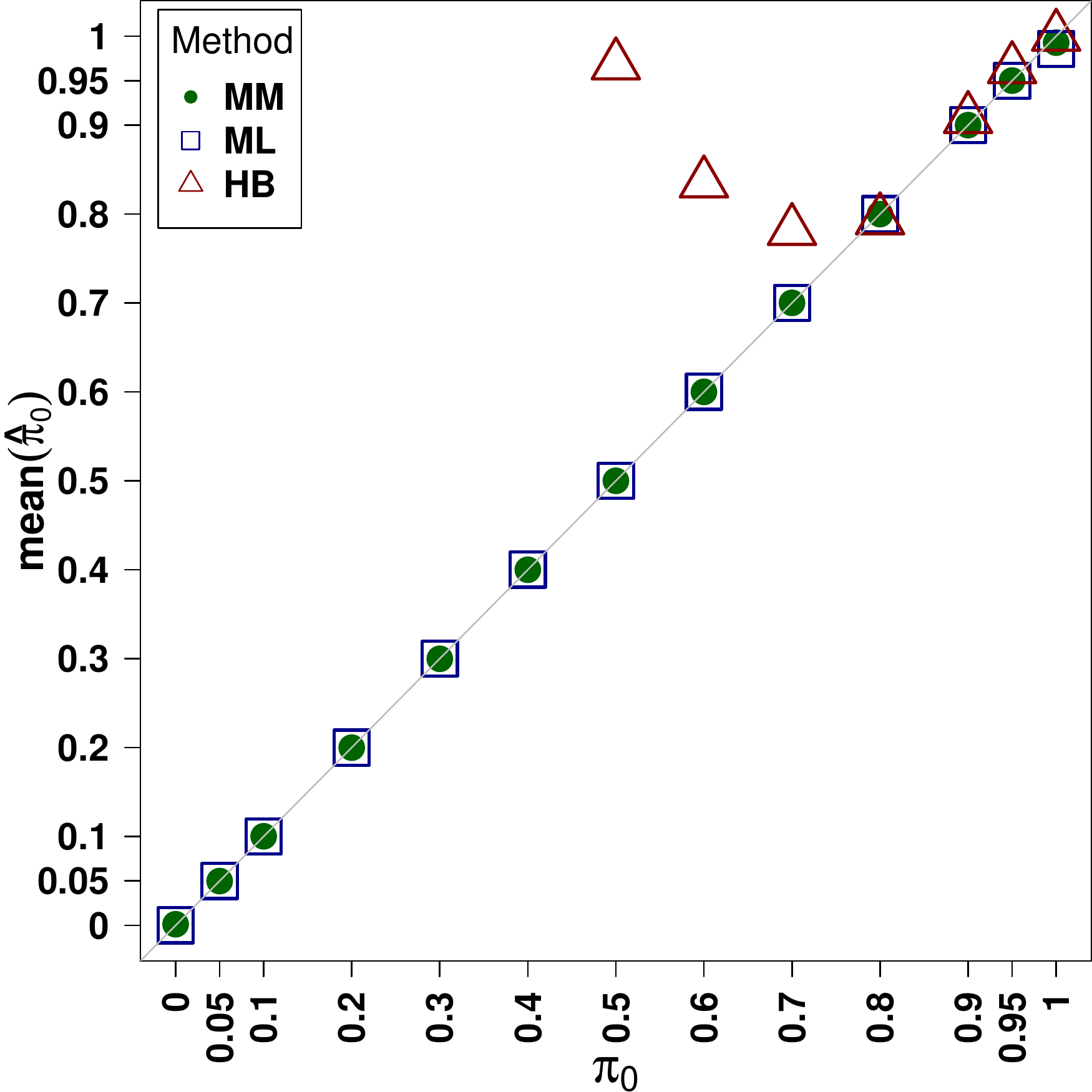}
}
\subfloat[]{
  \includegraphics[width=60mm]{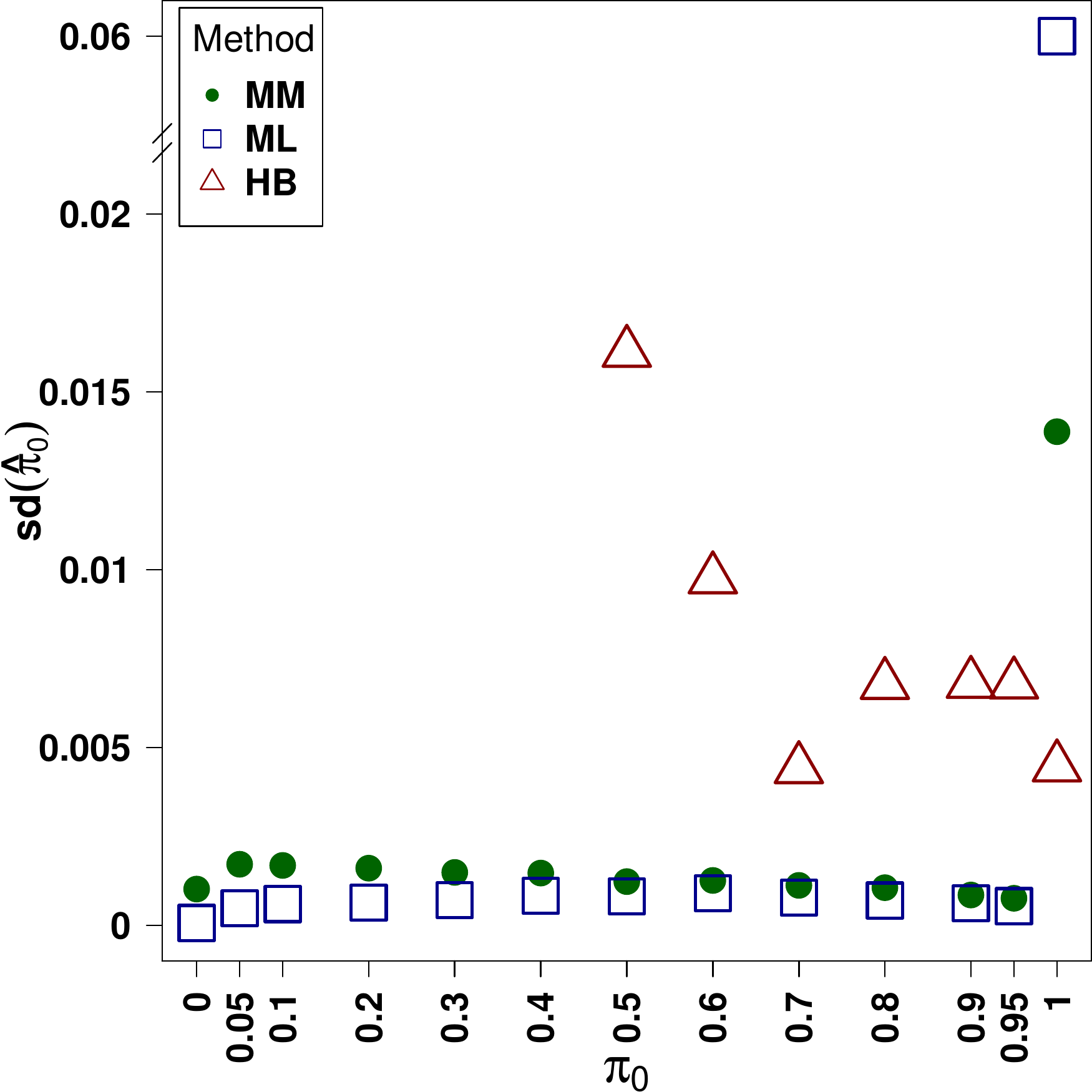}
}
\caption{Plots of (a) mean($\widehat{\pi}_0$), (b) sd($\widehat{\pi}_0$), computed based on the MM, ML and HB approaches for different values of $\pi_{0}$ in Algorithm \ref{alg1}. The gray line in panel (a) represents the identity line.}\label{fig:mean_sd_pi0}
\end{figure}
 \begin{figure}
\centering
\subfloat[]{
  \includegraphics[width=60mm]{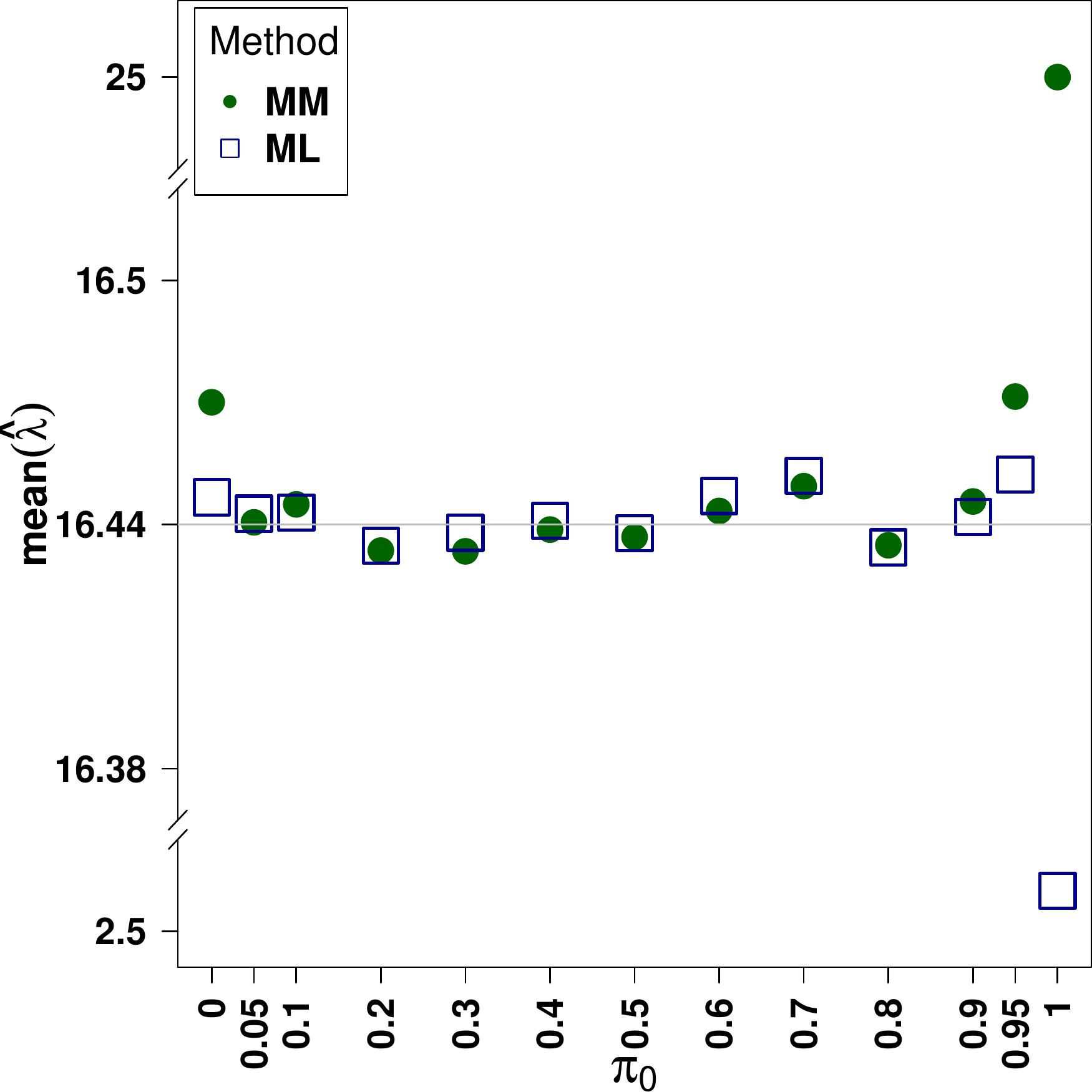}
}
\subfloat[]{
  \includegraphics[width=60mm]{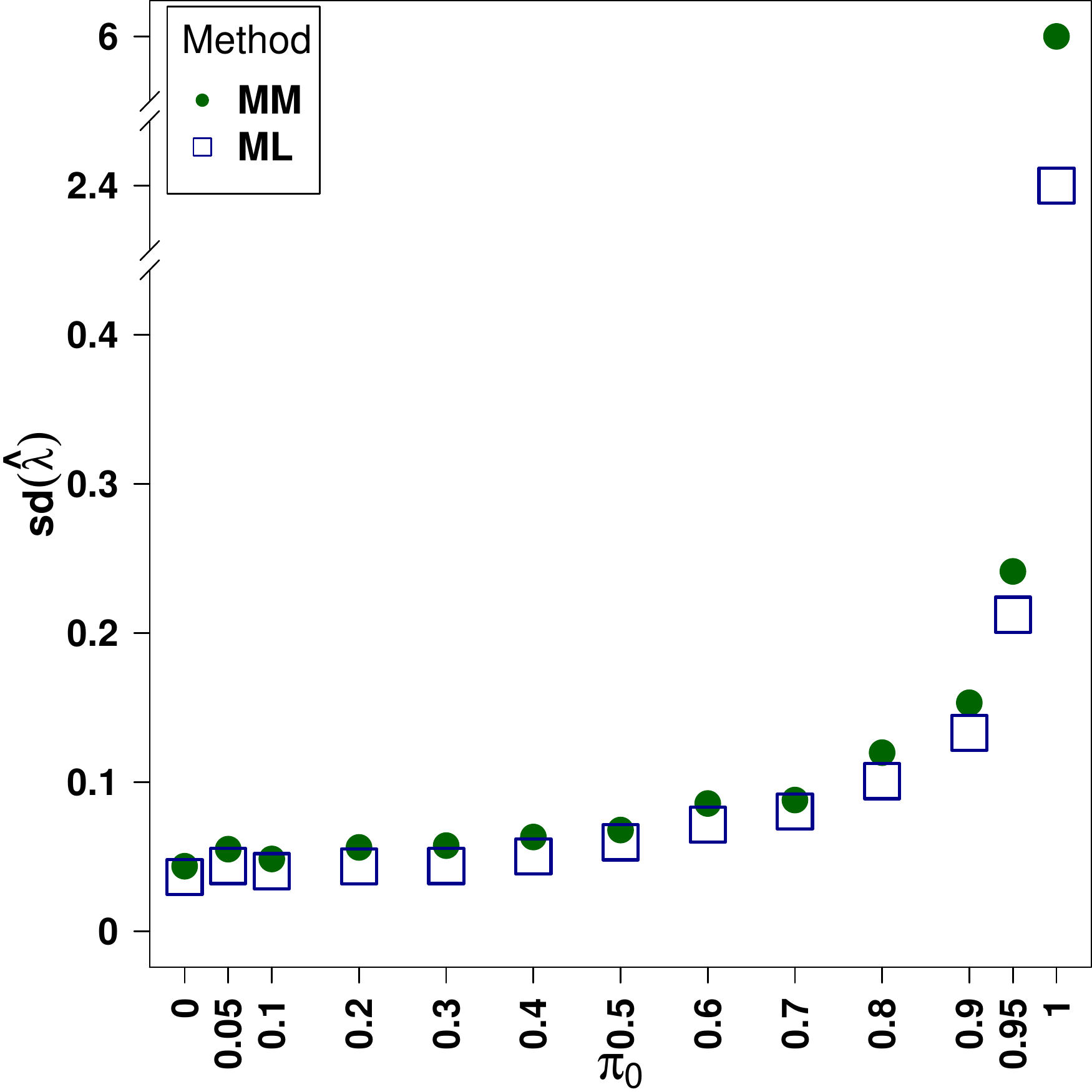}
}
\caption{Plots of (a) mean($\widehat{\lambda}$), (b) sd($\widehat{\lambda}$), computed based on the MM and ML approaches for different values of $\pi_{0}$ in Algorithm \ref{alg1}. The gray line in panel (a) represents the true $\lambda$.}\label{fig:mean_sd_lambda}
\end{figure}

 \begin{figure}
\centering
\subfloat[]{
  \includegraphics[width=60mm]{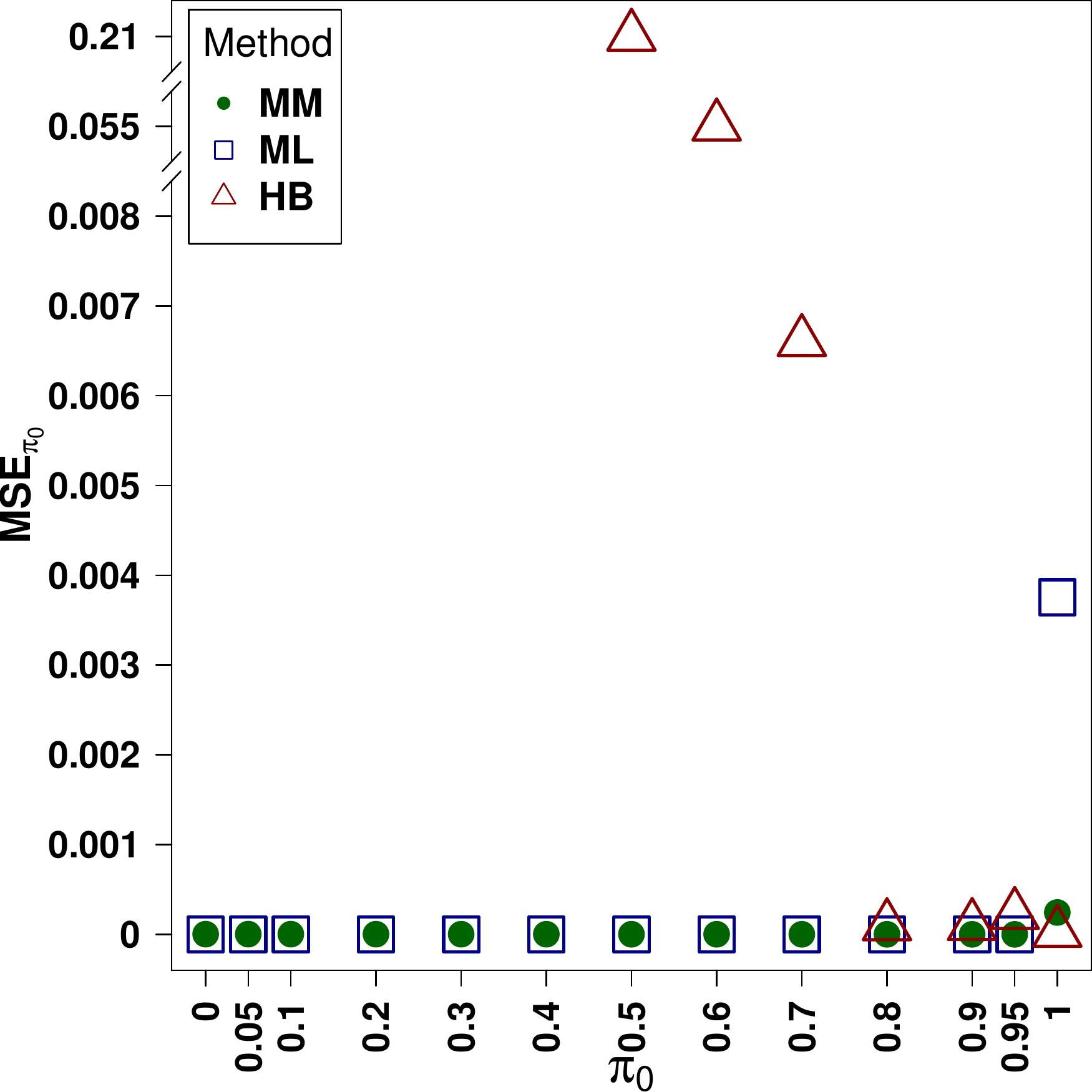}
}
\subfloat[]{
  \includegraphics[width=60mm]{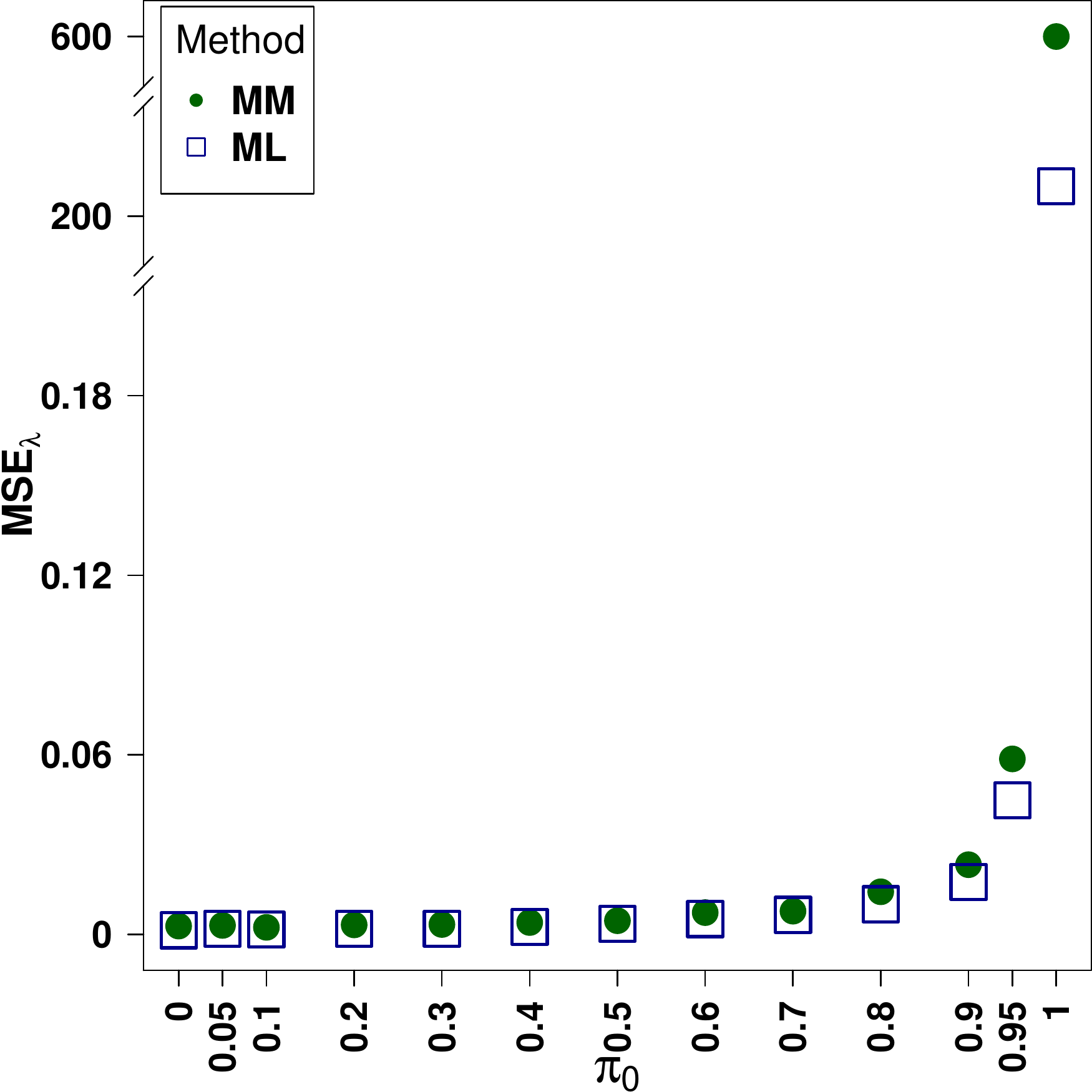}
}
\hspace{0mm}
\subfloat[]{
  \includegraphics[width=60mm]{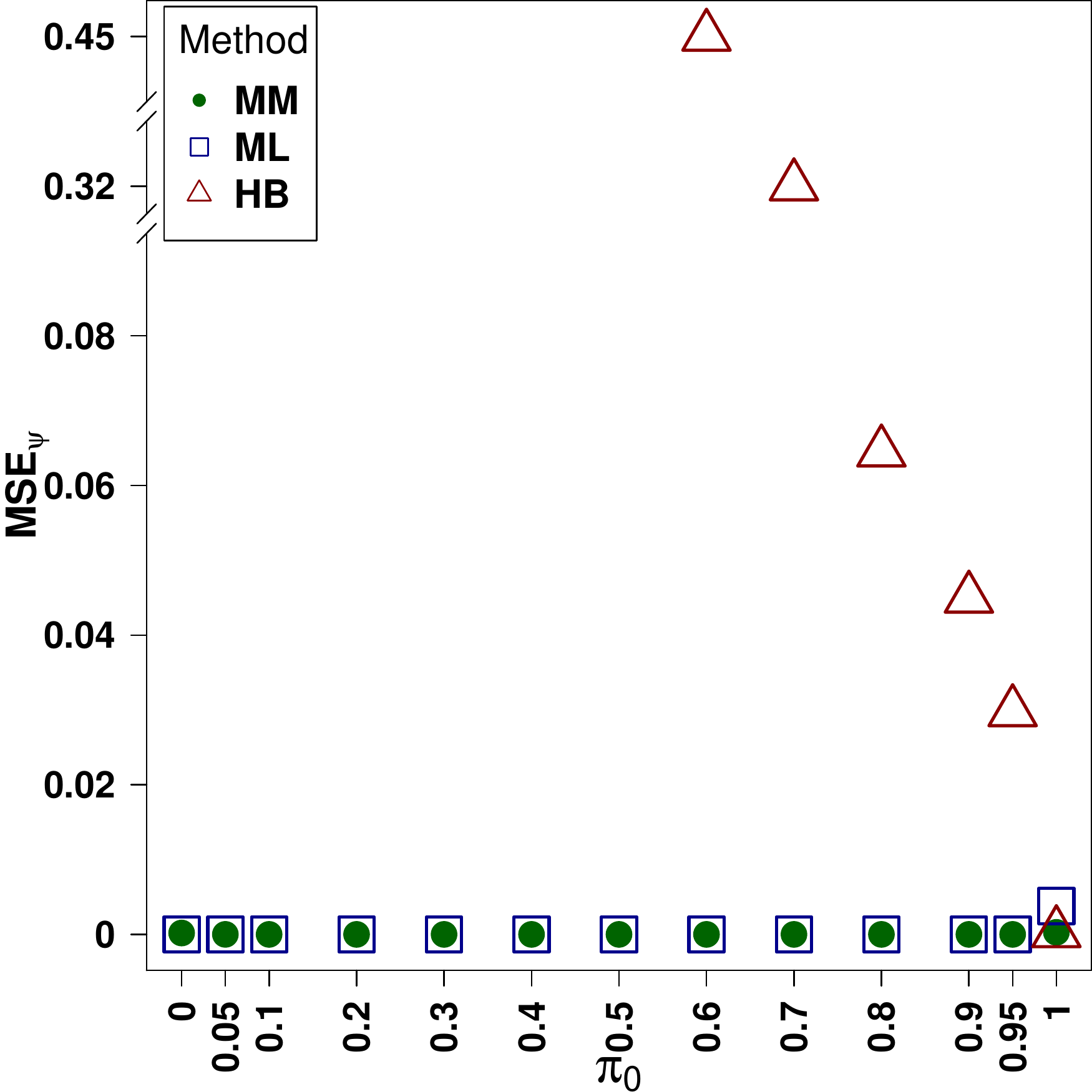}
}
\subfloat[]{
  \includegraphics[width=60mm]{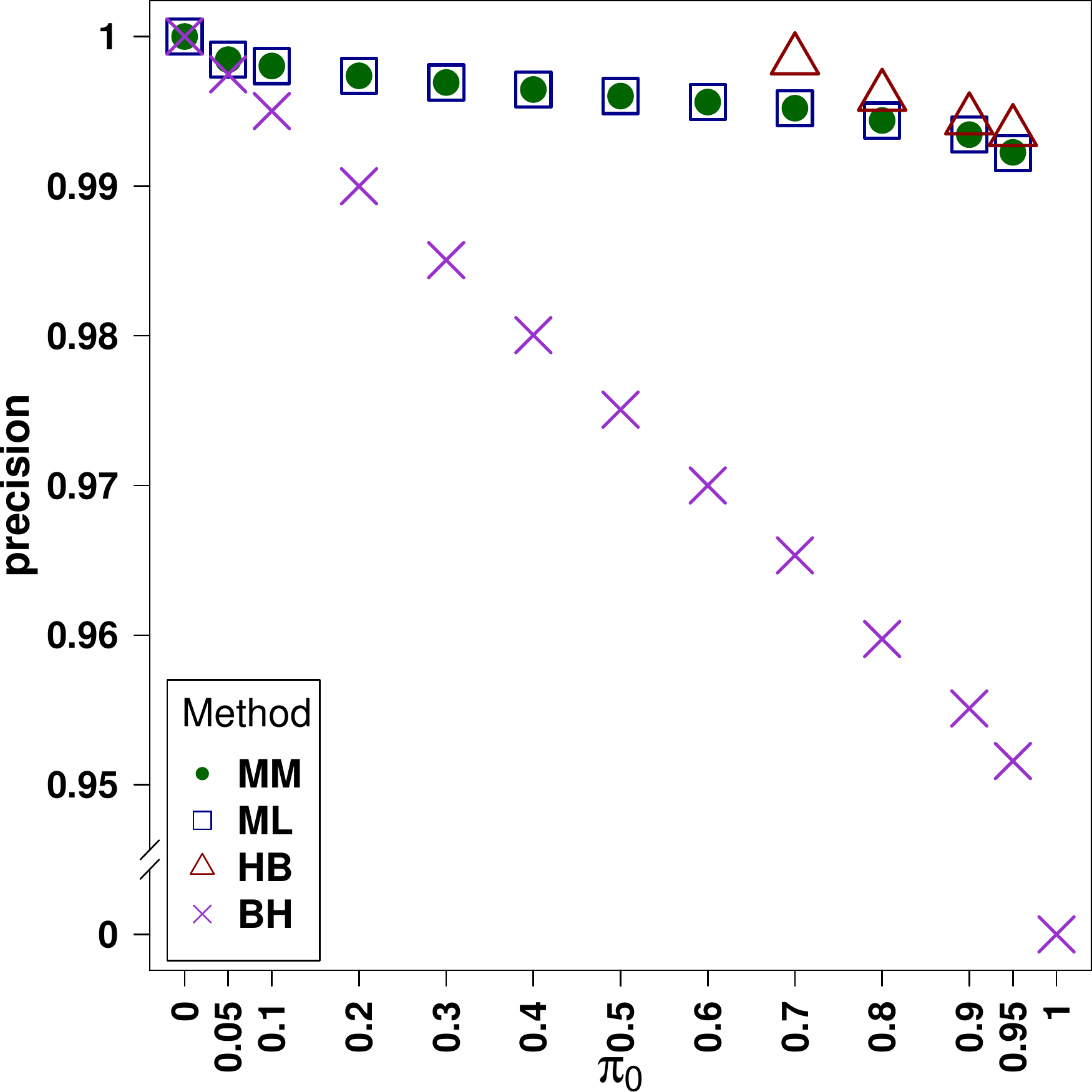}
}
\caption{Plots of (a) $MSE_{\pi_0}$, (b) $MSE_{\lambda}$, (c) $MSE_{\psi}$, (d) precision, for different values of $\pi_{0}$ in Algorithm \ref{alg1}.}\label{fig:MSE_pi0-and-MSE_lambda-and-MSE_psi}
\end{figure}

Figure \ref{fig:MSE_pi0-and-MSE_lambda-and-MSE_psi}(a)-(c) reflects $MSE_{\pi_0}$, $MSE_\lambda$ as well as $MSE_\psi$ at different selected true $\pi_0$ values. From Figure \ref{fig:MSE_pi0-and-MSE_lambda-and-MSE_psi}(a), we observe that the HB approach led to the highest MSE  when $\pi_0< 0.8$. It performed well for other values of $\pi_0$. The ML and MM approaches performed very well at all $\pi_0$ levels except when $\pi_0=1$. At  $\pi_0=1$, the HB approach performed better than the other two approaches and the MM approach outperformed the ML one. From Figure \ref{fig:MSE_pi0-and-MSE_lambda-and-MSE_psi}(b), it is observed that both the MM and ML approaches performed very well when $\pi\leq0.9$. Their performance is still satisfactory for other values of $\pi_0$. However, the ML approach led to lower MSE values than the MM approach. Figure \ref{fig:MSE_pi0-and-MSE_lambda-and-MSE_psi}(c) reflects values of MSEs in estimating LFDRs. From this figure it is observed that the HB approach led to the highest MSE values at all levels of $\pi_0$ except when $\pi_0=1$. Both the MM and ML approaches performed very well at all levels of $\pi_0$. At $\pi_0=1$, the MM and HB approaches outperformed the ML approach. It is interesting to add that comparing Figures  \ref{fig:mean_sd_pi0},  \ref{fig:mean_sd_lambda} and \ref{fig:MSE_pi0-and-MSE_lambda-and-MSE_psi}(a)-(c) reveals that the performance of both MM and ML approaches depends on accuracy of estimation in $\pi_0$ rather than $\lambda$. For example at $\pi_0=1$, from Figure \ref{fig:MSE_pi0-and-MSE_lambda-and-MSE_psi}(a) and (c), we observe that the MM approach led to a lower MSE value than the ML approach while from Figure \ref{fig:MSE_pi0-and-MSE_lambda-and-MSE_psi}(b), it led to a higher MSE value than the ML approach.

We also compared the performance of the proposed MM method with the BH method, as one of classic approaches existing in the literature. \textcolor{Reviewer1}{For a given FDR level $\alpha$ and a set of $N$ ascending p-values $p_{(1)},\ldots,p_{(N)}$, the BH method finds the largest index, say $k$, for which $p_{(k)}\leq \frac{i}{N} \alpha$}, and then $H_{01},\ldots,H_{0k}$ are rejected. As a measure of performance, we calculated the rate of true discoveries over the sum of true discoveries and false discoveries. We call this rate ``precision". \textcolor{Reviewer1}{Note that precision in known to be equal to $1-FDR$}. A true discovery here means rejecting the null hypothesis while it is not true, and similarly, a false discovery means rejecting the null hypothesis while it is true. \textcolor{Reviewer1}{To recognize whether a true discovery occurs, we compare the estimated LFDR value with a threshold, say 0.05.} Thus, with the settings of the simulations in Algorithm \ref{alg1}, if for $i=1,\ldots,N_0$, $\widehat{\psi}_i\leq 0.05$, we deduce that a true discovery occurs. Also, if for $i=N_0+1,\ldots,N$, $\widehat{\psi}_i\leq 0.05$, we infer that a false discovery occurs. Thus, we measure precision by  $\frac{\sum_{j=1}^{N_0}I(\widehat{\psi}_i^M\leq 0.05)}{\sum_{j=1}^{N}I(\widehat{\psi}_i^M\leq 0.05)}$, where $M\in\{MM,ML,HB,BH\}$. Figure \ref{fig:MSE_pi0-and-MSE_lambda-and-MSE_psi}(d) represents precision calculated at difference levels of true $\pi_0$. From the figure we observe that when $\pi_0\in[0.7,1)$, the HB approach led to the highest precision values compared to the MM, ML and BH approaches (it reported 0 over 0 at $\pi_0=0.5,0.6,1$, and failed to return output when $\pi_0<0.5$). Both the MM and ML approaches led to satisfactory precision values at all levels of $\pi_0$. The MM, ML and HB methods all returned 0 over 0 at $\pi_0=1$, as expected (at $\pi_0=1$, an ideal LFDR estimation method is expected to lead to no discoveries) but the BH approach led to a 0 precision, which means it identified some false discoveries. \textcolor{Reviewer1}{The figure also represents that the BH approach successfully controlled FDR to the 95\% (or $1-\alpha$) level, as expected.}

\begin{figure}[b]
\centering
\subfloat[]{
  \includegraphics[width=60mm]{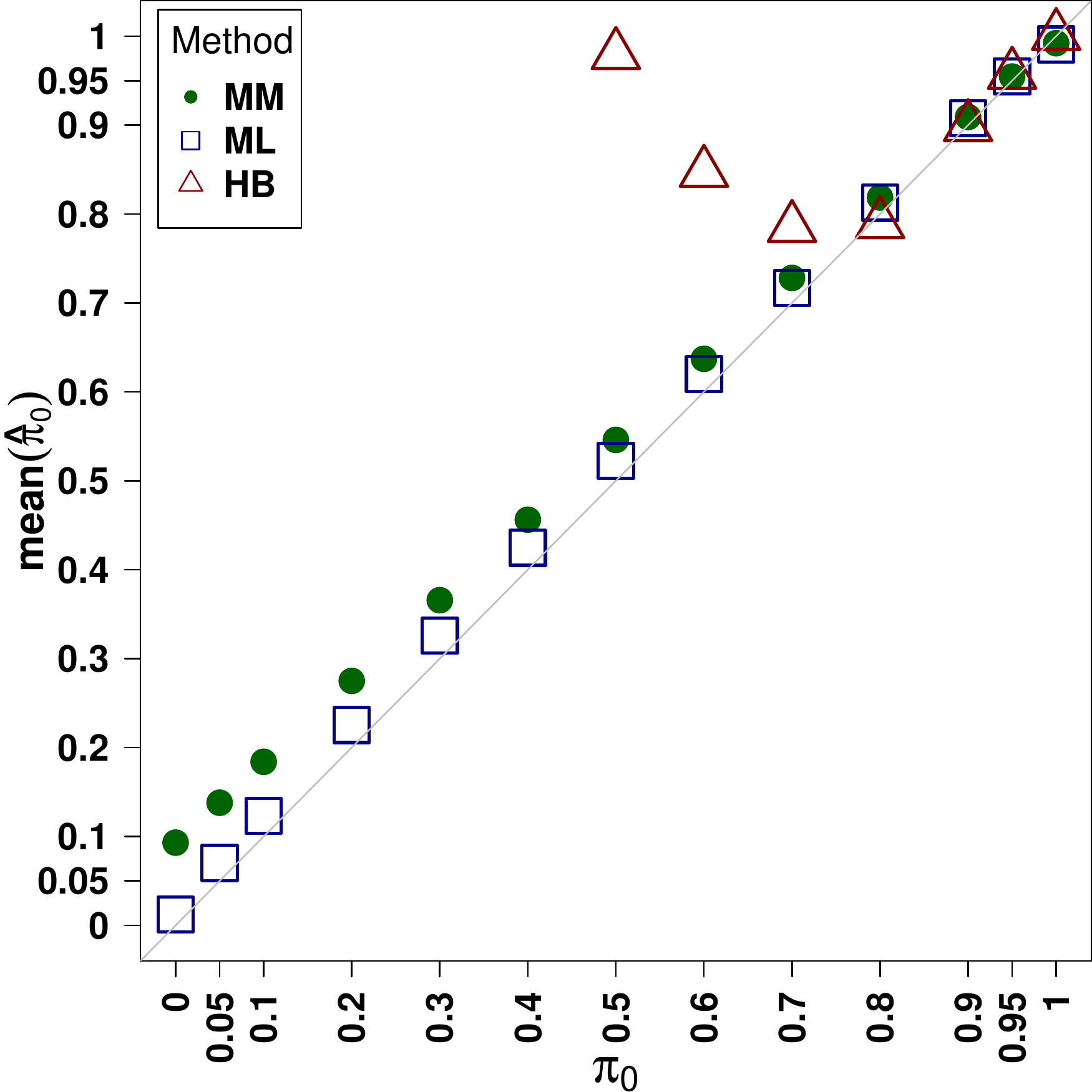}
}
\subfloat[]{
  \includegraphics[width=60mm]{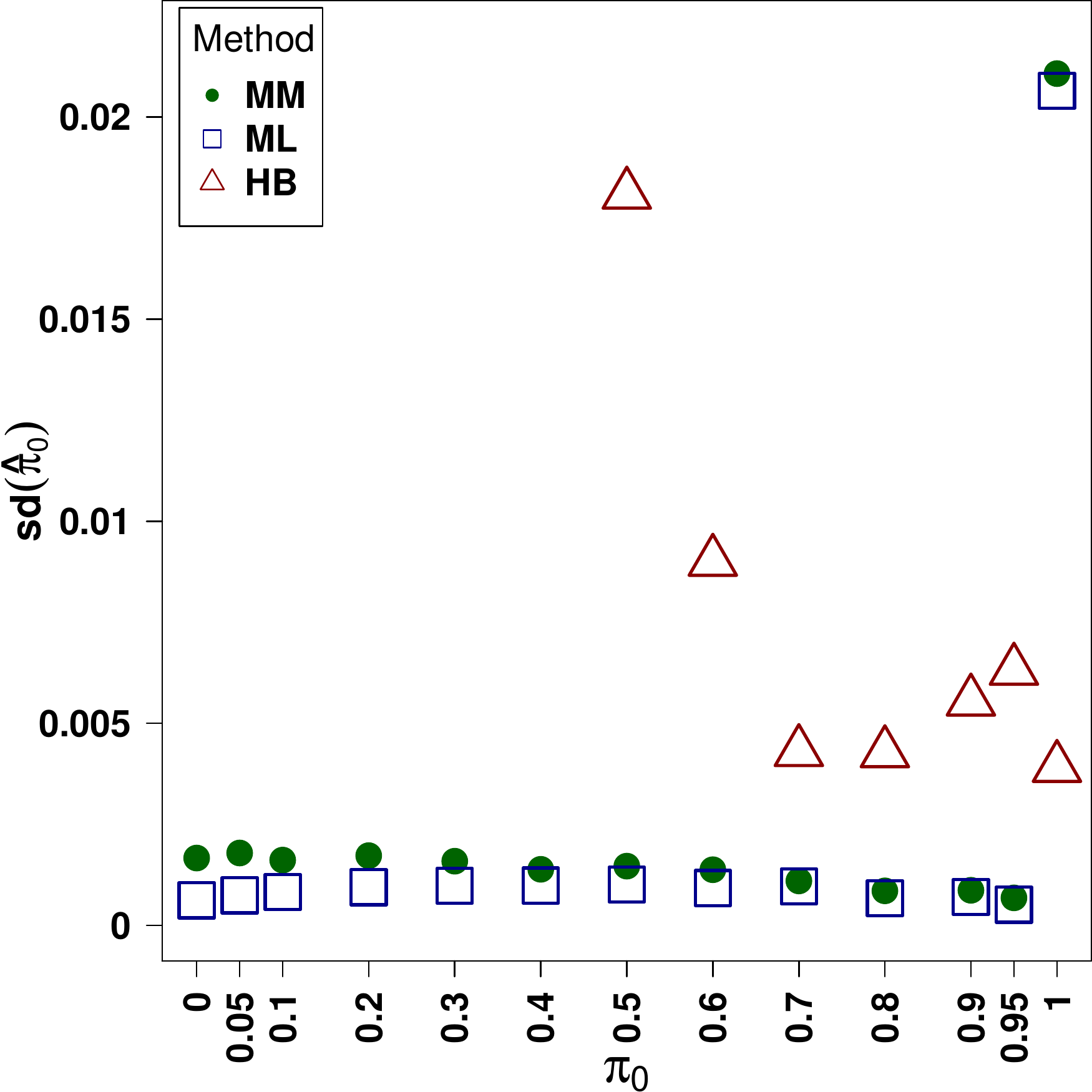}
}
\caption{Plots of (a) mean($\widehat{\pi}_0$), (b) sd($\widehat{\pi}_0$), computed based on the MM, ML and HB approaches for different values of $\pi_{0}$ in Algorithm \ref{alg1} with random ORs.The gray line in panel (a) represents the identity line.}\label{fig:mean_sd_pi0_OR}
\end{figure}

\begin{figure}
\centering
\subfloat[]{
  \includegraphics[width=60mm]{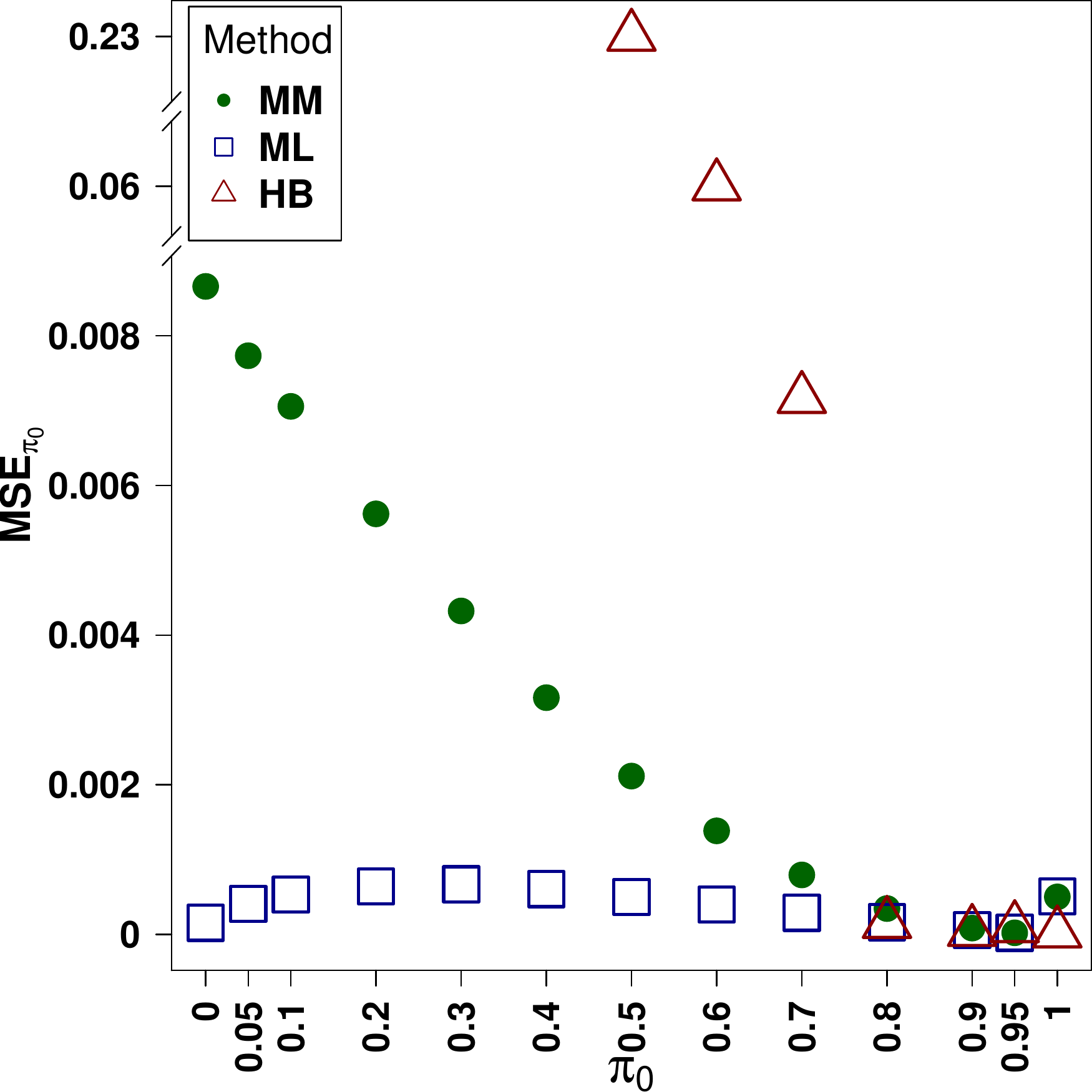}
}
\subfloat[]{
  \includegraphics[width=60mm]{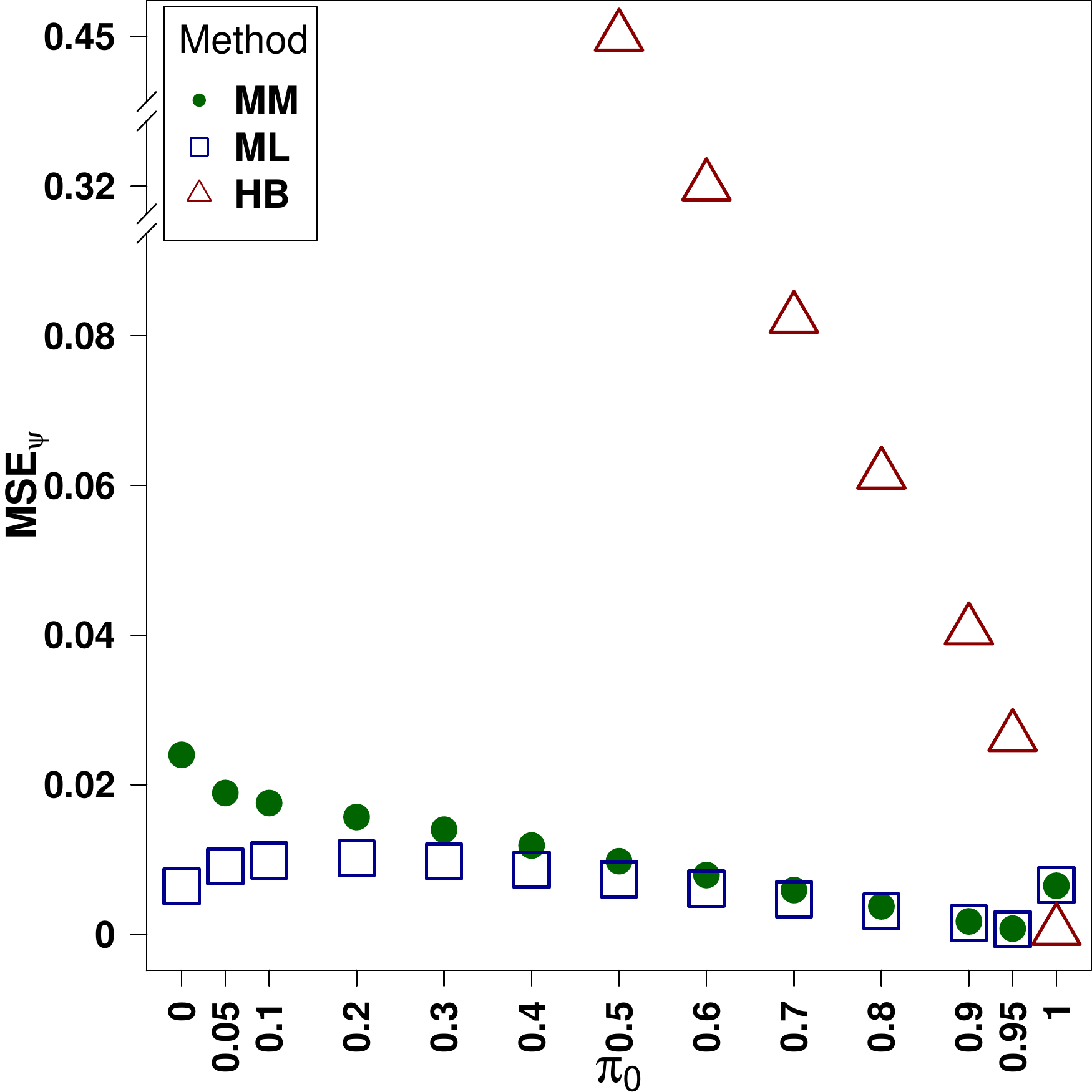}
}
\hspace{0mm}
\subfloat[]{
  \includegraphics[width=60mm]{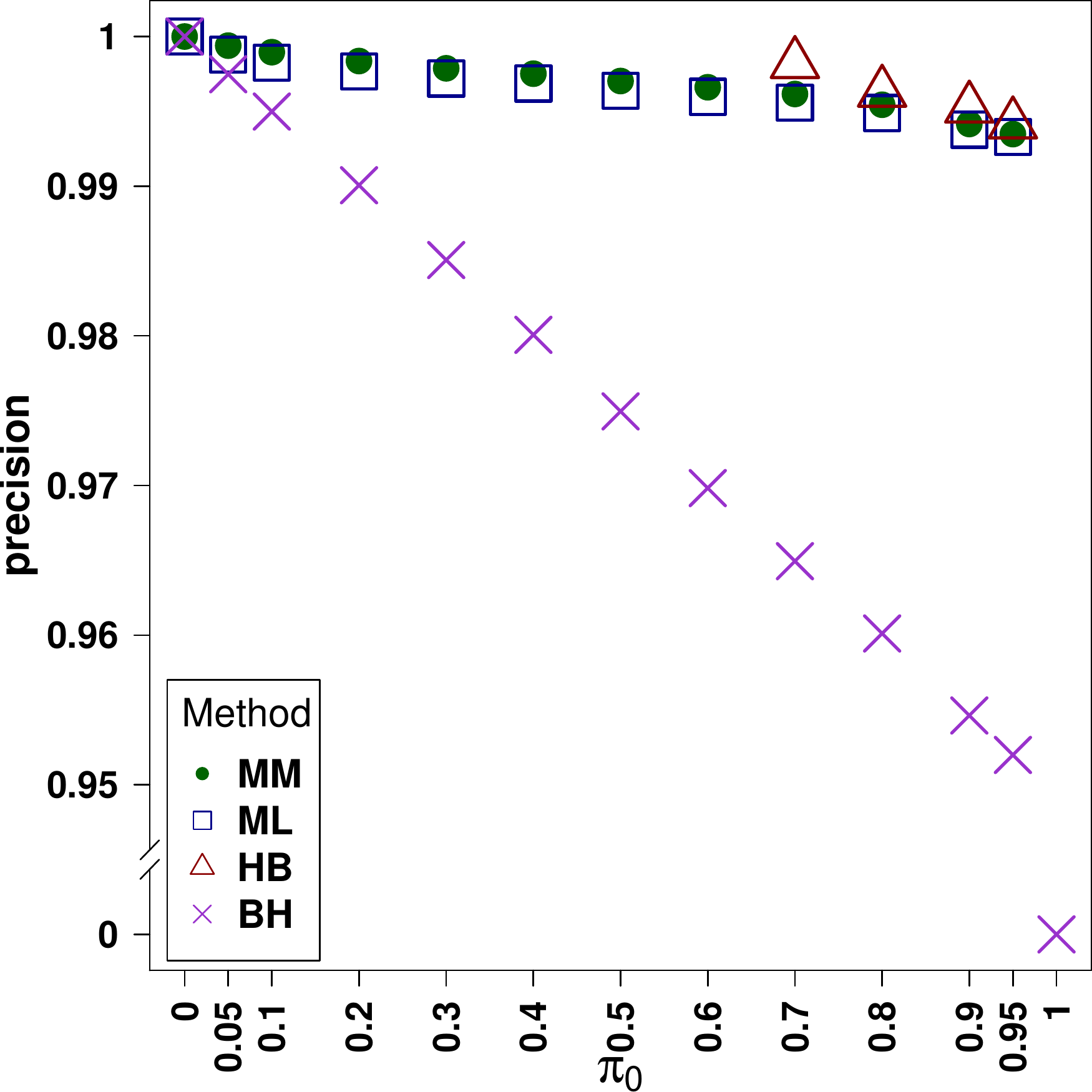}
}
\caption{Plots of (a) $MSE_{\pi_0}$, (b) $MSE_{\psi}$, (c) precision, for different values of $\pi_{0}$ in Algorithm \ref{alg1} will random ORs.}\label{fig:MSE_pi0-and-MSE_psi_random}
\end{figure}

We also evaluated robustness of the proposed method w.r.t. the violation of the same non-centrality parameter assumption. In this regard, instead of having the same OR for associated SNPs in Algorithm \ref{alg1}, we assumed that OR follows  a $N(1.5,0.1)$-distribution, and thus generated different OR values for each associated SNP. We then followed the steps in Algorithm  \ref{alg1}. Figures \ref{fig:mean_sd_pi0_OR} and \ref{fig:MSE_pi0-and-MSE_psi_random} represent simulation results. Note that we were unable to produce any plot for either $\widehat{\lambda}$ or $MSE_{\lambda}$, as the true $\lambda$  for each associated SNP was a random number.  \textcolor{Reviewer1}{Comparing Figure \ref{fig:mean_sd_pi0_OR} with Figure \ref{fig:mean_sd_pi0}, it is observed that for $\pi_0\leq0.8$, the performance of both the MM and ML approaches got worse and, the ML approach performed a bit better than the MM one. However, their performance did not change when $\pi_0>0.8$. The performance of the HB method had no significant changes.} Comparing Figure \ref{fig:MSE_pi0-and-MSE_psi_random}(a) with Figure \ref{fig:MSE_pi0-and-MSE_lambda-and-MSE_psi}(a), it is clear that the performance of both the MM and ML approaches got worse when $\pi_0\leq0.8$, although the ML approach performed better than the MM one, and when $\pi_0>0.8$, both the MM and ML approaches performed the same.  Again, the performance of the HB method did not change. The same conclusion is made when comparing Figure \ref{fig:MSE_pi0-and-MSE_psi_random}(b) with Figure \ref{fig:MSE_pi0-and-MSE_lambda-and-MSE_psi}(c). Comparing Figure \ref{fig:MSE_pi0-and-MSE_psi_random}(c) with Figure \ref{fig:MSE_pi0-and-MSE_lambda-and-MSE_psi}(d), we observe almost the same pattern with very small changes in precision values of the MM and ML methods. From the above performance analysis, we conclude that when $\pi_0\leq0.8$, replacing common non-centrality parameters by uncommon ones impacts the performance of both the ML and MM methods. However, the performance difference is negligible when $\pi_0>0.8$. We should emphasize that $\pi_0$ in practice is greater than 0.9, and thus, the assumption of having the same non-centrality parameter should not significantly impact  analysis results.   

\textcolor{reviewer1}{Summarizing the above performance evaluation, we conclude that, compared to the ML, HB and BH methods, the MM method performs well and the resulting LFDR estimates are highly precise and reliable.}

\subsection{Second simulation study}\label{subsec:5.2}
We simulate case-control samples for each SNP given an additive model. The simulation strategy follows the steps in Algorithm \ref{alg2}. Although this simulation strategy includes more parameters than the ones used in Algorithm \ref{alg1}, it does not allow to control the true parameter $\lambda$. The only true parameter which is known from the beginning of the simulation is $\pi_0$. Thus, in this simulation strategy, we are only able to measure the accuracy in estimating $\pi_0$. However, since $\widehat \lambda$ in equation (\ref{p0.hat}) directly depends on the value of $\widehat\pi_0$, a perfect estimate of  $\pi_0$ would automatically lead to a reliable estimate of $\lambda$.

Following Algorithm \ref{alg2}, we conducted different simulations with different parameters. \textcolor{reviewer1}{We took $r=60,000$, $s=120,000$, $p=0.2$, $v_0=0.01$, $OR_2=1.5$, $N=1,000,000$,  \textcolor{Reviewer2}{$\pi_0=0, 0.05$, $0.10(0.1)0.90, 0.95, 1$} and $b=100$}. \textcolor{Reviewer1}{Figure \ref{fig:sim2_mean_sd_pi0}(a)-(b) represents plots of mean and sd of $\widehat{\pi}_0$ at different selected true $\pi_0$ values. Also, Figure \ref{fig:sim2_MSE_precision}(a)-(b) displays $MSE_{\pi_0}$ and precision at the same $\pi_0$ levels. From Figure  \ref{fig:sim2_mean_sd_pi0}(a)-(b), we observe that, the HB method failed when $\pi_0<0.5$. It estimated $\pi_0$ very well only when it is very close to 1. According to Figure  \ref{fig:sim2_mean_sd_pi0}(a), the ML and MM approaches performed well only when $\pi_0\in[0.5,0.9]$. Their performance was weak for other $\pi_0$ values. When $\pi_0<0.5$, both the MM and ML methods over-estimated $\pi_0$, and when $\pi_0>0.5$ both methods under-estimated it. We also observe from Figure  \ref{fig:sim2_mean_sd_pi0}(b) that when $\pi_0\leq 0.5$, the ML approach led to lower variance values than the MM approach. However, the MM estimator led to a variance improvement for the rest of $\pi_0$ values.} 
Almost the same behavior is observed in Figure \ref{fig:sim2_MSE_precision}(a). When $\pi_0\geq0.9$, the HB approach performed very well, and the MM approach outperformed the ML one. From Figure \ref{fig:sim2_MSE_precision}(b), we observe that the HB approach led to the highest precision values compared to the MM, ML and BH approaches (again, it reported 0 over 0 at $\pi_0=0.5,0.6,0.7,1$, and failed to return output when $\pi_0<0.5$). The MM and ML approaches performed similarly. The BH approach led to higher precision than the MM and ML approaches when $\pi_0<0.5$. However, the MM and ML approaches outperformed the BH one when $\pi_0>0.5$. All methods except the BH method returned 0 over 0 at $\pi_0=1$ while the BH one led to a 0 precision. \textcolor{Reviewer1}{Figure \ref{fig:sim2_MSE_precision}(b) also reveals that the BH approach successfully controlled FDR to the 95\% level again.} Overall, it is concluded that the proposed method leads to satisfactory results.

\begin{algorithm}
\caption{Second simulation strategy.}\label{alg2}
{\begin{enumerate}
\item[] Step1. Specify the numbers of cases ($r$) and controls ($s$), the allele frequency $p$ for the risk \\ \hspace*{11mm}allele $B$ and the reference penetrance $v_0$.
\item[]  Step 2. Take $l=1$.
\item[]  Step 3. Specify $OR_2\neq 1$.
\item[]  Step 4. Calculate $v_2$ using the following equation
\[v_2= \frac{exp(\beta_0+\beta_2)}{1+exp(\beta_0+\beta_2)},\]
\hspace*{11mm}where $\beta_0=\log\left(\frac{v_0}{1-v_0}\right)$ and $\beta_2=\log(OR_2)$ (the above equation is in fact the\\ \hspace*{11mm}prospective logistic regression model).  
\item[] Step 5. Calculate $v_1=\frac 1 2 (v_0+v_2)$ (this is due to selecting ad additive model).
\item[] Step 6. Calculate $k=\sum_{j=0}^2 v_j g_j$, where $g_0=(1-p)^2$, $g_1=2p(1-p)$ and $g_2=p^2$.
\item[]  Step 7. For $j=1,2,3$, calculate $p_j =g_jv_j/k$ and $q_j =g_j(1-v_j)/(1-k)$.
\item[]  Step 8. Take $m=1$.
\item[]  Step 9. Generate random samples $(r_0,r_1,s_2)$ and $(s_0,s_1,s_2)$ independently from the
multi-\\ \hspace*{11mm}nomial distributions $Mul(r;p_0, p_1, p_2)$ and $Mul(s; q_0, q_1, q_2)$, respectively. This \\ \hspace*{11mm}leads to a $2\times3$ Table similar to Table \ref{Tab1}.
\item[]  Step 10. Similar to Table \ref{Tab2}, construct the corresponding $2\times2$ Table and compute the chi-\\ \hspace*{11mm}square test statistic of independence, i.e., $x_l=\sum_{j=1}^{4}(o_j-e_j)^2/e_j$, where\\ \hspace*{11mm} $o_1=2r_0+r_1$, $o_2=r_1+2r_2$, $o_3=2s_0+s_1$, $o_4=s_1+2s_2$, \\ \hspace*{12mm}$e_1=2R(2n_0+n_1)/(2(R+S))$, $e_2=2R(n_1+2n_2)/(2(R+S))$,\\ \hspace*{11mm} $e_3=2S(2n_0+n_1)/(2(R+S))$, $e_4=(2S)(n_1+2n_2)/(2(R+S))$\\ \hspace*{11mm} with $R=r_0+r_1+r_2$, $S=s_0+s_1+s_2$ and $n_j=r_j+s_j$, $j=1,2,3$. 
\item[] Step 11. Step up $m$ by one and repeat Steps 9-10 until $m={N_0}$.
\item[] Step 12. Take $OR_2=1$ and repeat Steps 4-7. 
\item[] Step 13. Increase $m$ by one, and repeat Steps 9-10 until $m=N$.
\item[] Step 14. Estimate $\pi_0$  by $\widehat{\pi}_0^M$, where $M$ indicates one of the MM, ML and HB approaches.
\item[] Step 15. Compute the error of estimating the true proportion of unassociated SNPs by\\ \hspace*{11mm} $E_{\pi_0}^{l,M}=(\widehat{\pi}_0^M-\pi_0)^2$, where $\pi_0=1-\frac{N_0}{N}$.
\item[] Step 16. Increase $l$ by one and repeat Steps 3-15 for  \textcolor{reviewer2}{for $b$ times. Then, compute
\begin{eqnarray*}
MSE_{\pi_0}^M=\frac{1}{b}\sum_{l=1}^{b}E_{\pi_0}^{l,M}.
\end{eqnarray*}}
\end{enumerate}
}
\end{algorithm}

  \begin{figure}[h]
\centering
\subfloat[]{
  \includegraphics[width=60mm]{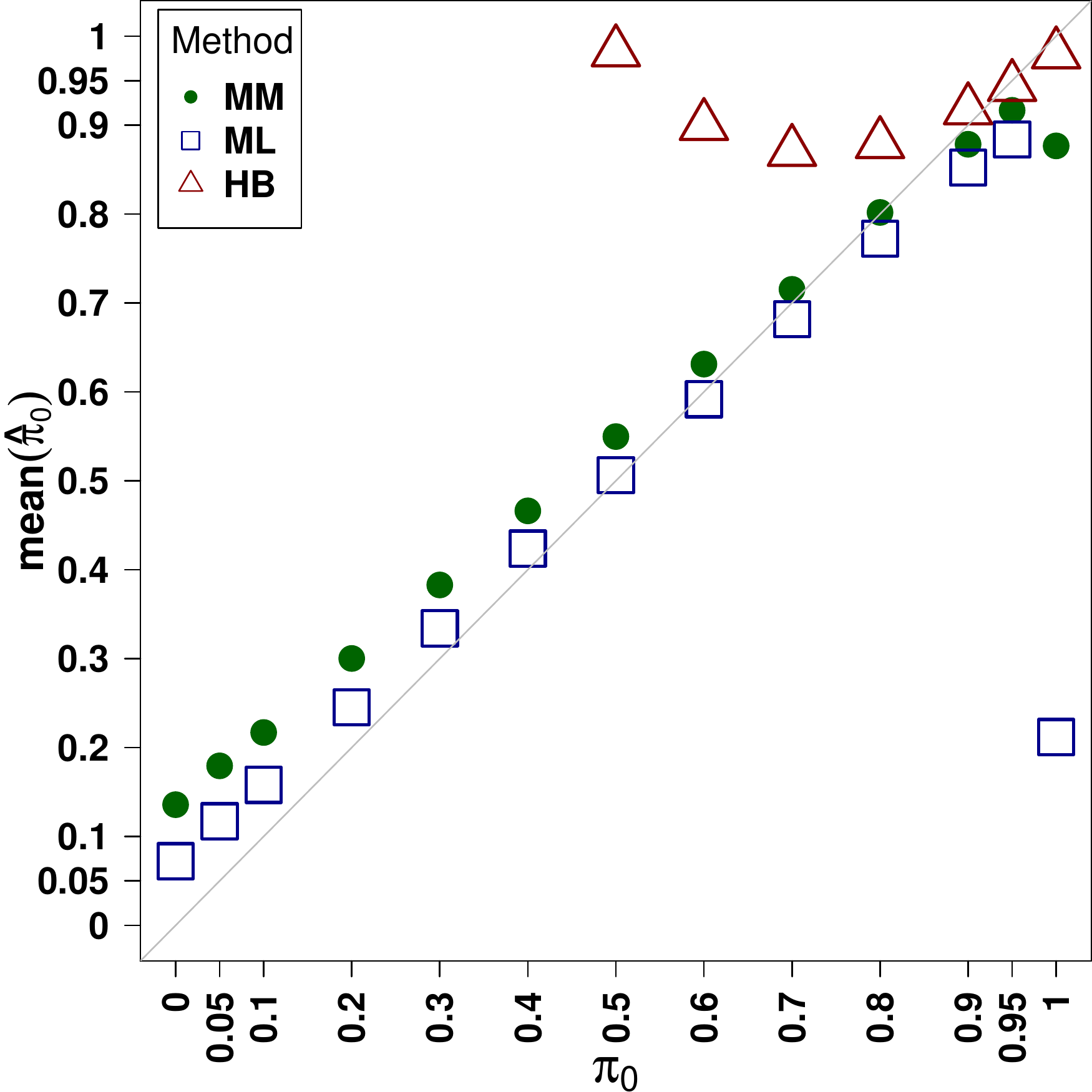}
}
\subfloat[]{
  \includegraphics[width=60mm]{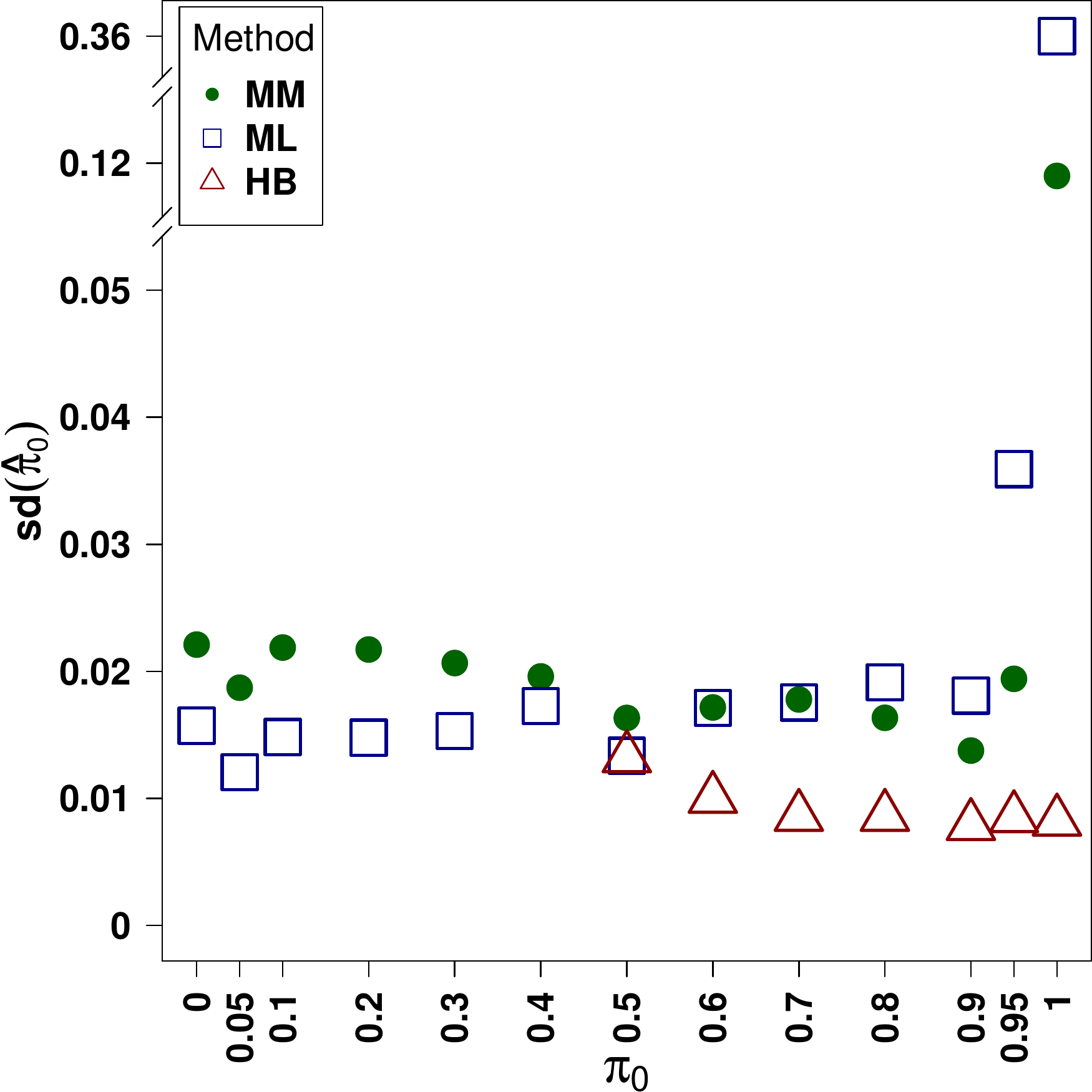}
}
\caption{(a) Plots of (a) mean($\widehat{\pi}_0$), (b) sd($\widehat{\pi}_0$), computed based on the MM, ML and HB approaches for different values of $\pi_{0}$ in Algorithm \ref{alg2}.}\label{fig:sim2_mean_sd_pi0}
\end{figure}

  \begin{figure}
\centering
\subfloat[]{
  \includegraphics[width=60mm]{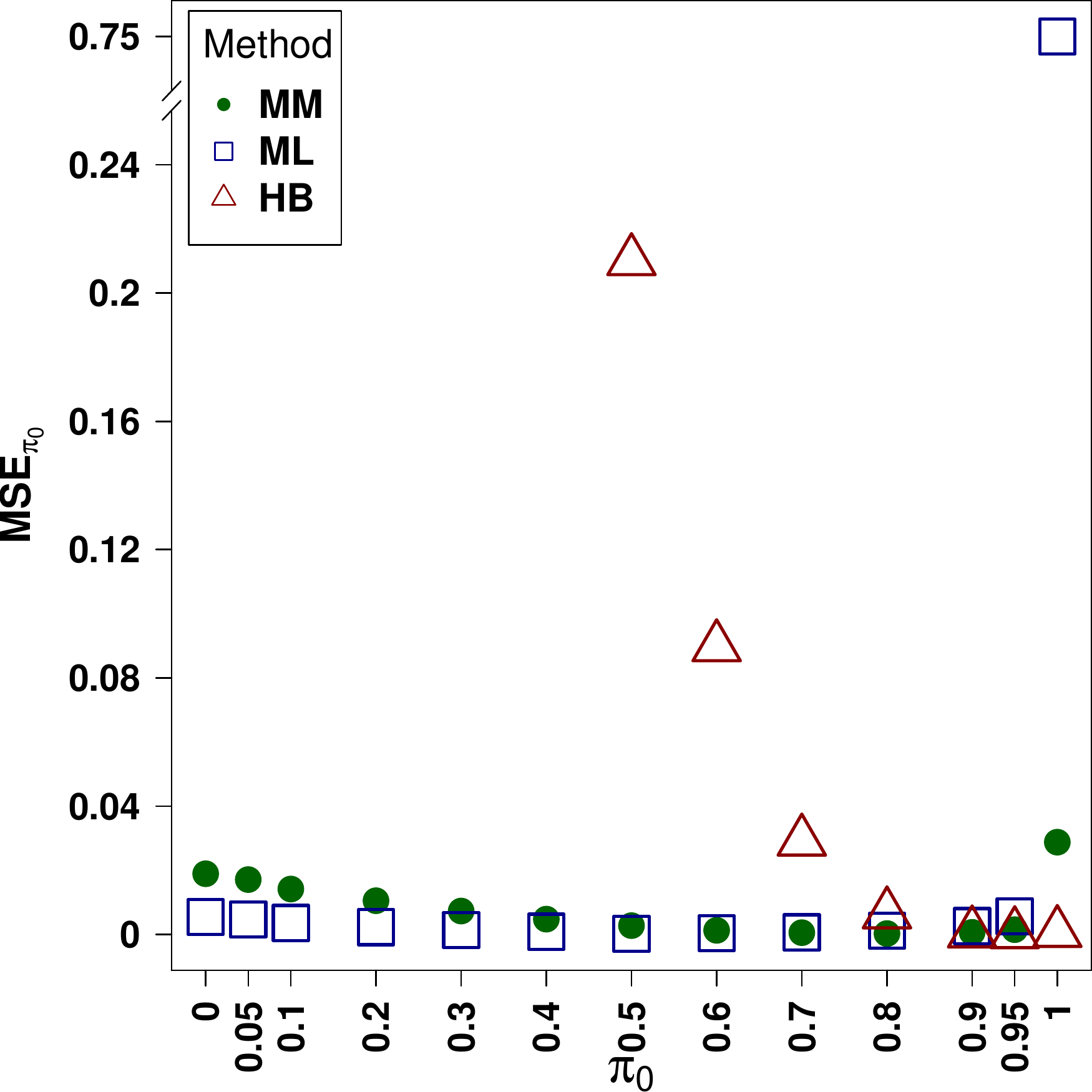}
}
\subfloat[]{
  \includegraphics[width=60mm]{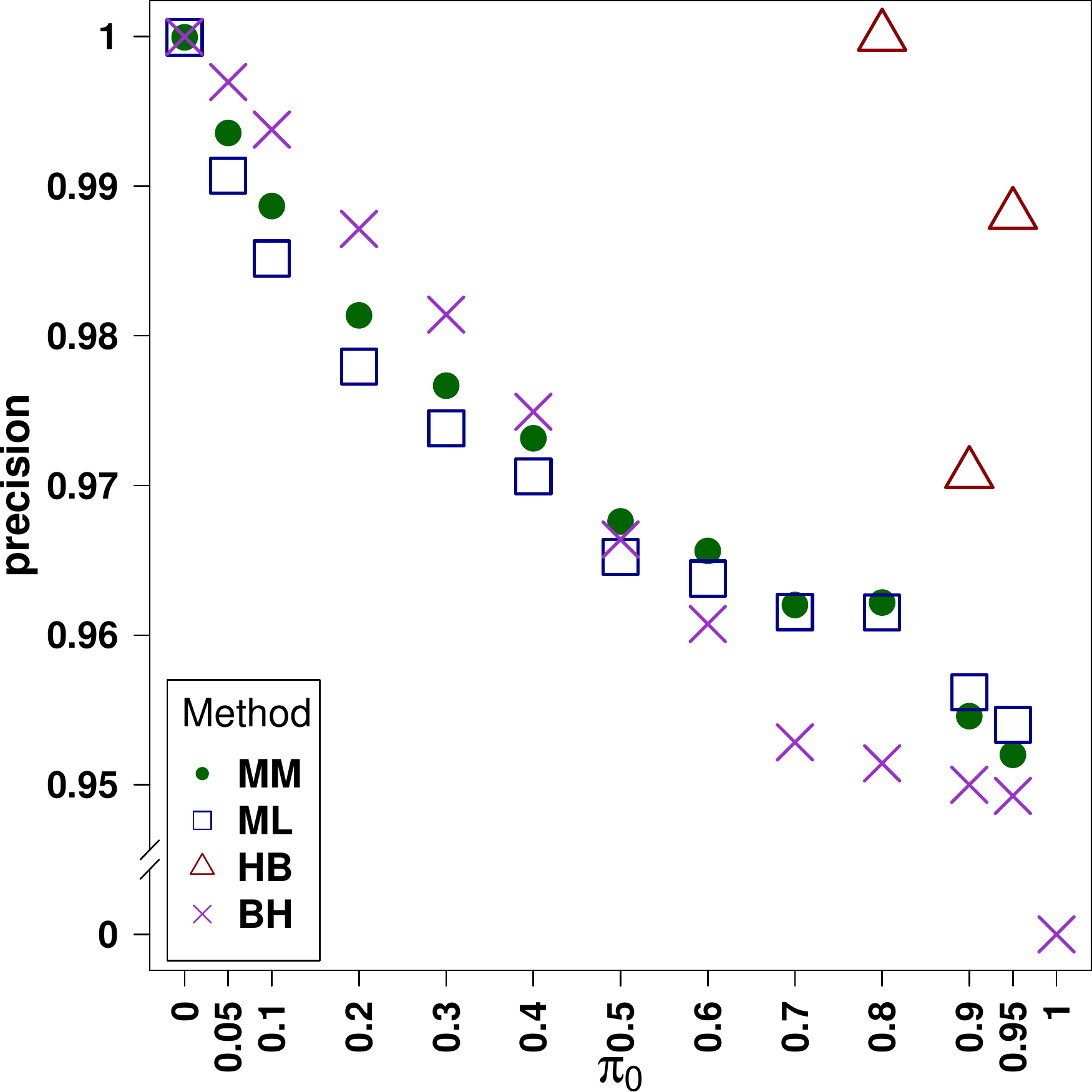}
}
\caption{(a) Plots of (a) $MSE_{\pi_0}$, (b) precision, for different values of $\pi_{0}$ in Algorithm \ref{alg2}.}\label{fig:sim2_MSE_precision}
\end{figure}

\section{Applications}\label{sec:6}

\subsection{Application to a comprehensive coronary artery disease data set}
In this subsection, we apply the proposed LFDR estimation approach to analyze a comprehensive 1000 genomes-based genome-wide association data, which was originally analyzed by \citet{Nikpay2015}.  \textcolor{reviewer2}{The data set consists of 60801 coronary artery disease (CAD) cases and 123504 controls,} and contains approximately 6.7 million variants with a minor allele frequency of greater than 0.05 and approximately 2.7 million variants with an allele frequency ranging between 0.005 and 0.05.  
The corresponding publicly available data consists of 9,455,777 SNPs with different information such as SNP name,  chromosome name, effect allele, non-effect allele, frequency of effect allele, logistic regression coefficient ($\widehat\beta$) with the corresponding standard deviation ($\widehat{\text se}_{\widehat\beta}$), and p-value. 

 \textcolor{reviewer2}{According to \citet{Nikpay2015}, a threshold of $5\times 10^{-8}$ for p-value led to 2,213 total variants to be significantly associated with CAD.  \citet{Nikpay2015} also using the BH approach reported that, the chosen p-value threshold corresponds to an FDR q-value of $2.1\times10^{-4}$. We observed in the previous section that such a classic approach is outperformed by the MM, ML and HB approaches. Here we are interested in verifying how the threshold of LFDR changes when alternative LFDR estimation methods are applied.}
 
 \textcolor{reviewer2}{To implement our proposed MM method, for each SNP $i$, we took the test statistic $x_i$ to be $\left(\frac{\widehat\beta_i}{\widehat{\text se}_{\beta_i}}\right)^2$. Then, equation (\ref{p0.hat}) with a processing time of 4.385 seconds on a personal computer (core i7, 3.5 GHz speed with 16 GB of RAM) led to $\widehat{\pi}_0= 0.9967$ and $\widehat{\lambda}= 21.9274$. By these estimates, estimated LFDRs for the 2,213 variants had a maximum of 0.0003. This suggests that there is strong evidence for the association of the 2,213 SNPs with the CAD. }
 
 \textcolor{reviewer2}{The HB approach led  to $\widehat{\pi}_0=0.9776$ (with a model misfit warning), and the ML approach with $c=0$ and $d=30$ in equation (\ref{ML_arg}), led to $\widehat{\pi}_0=0.9958$ and $\widehat{\lambda}_0=12.4034$. The HB and ML methods resulted in a maximum of 0.1288 and 0.0011 for estimated LFDRs correspondent to those 2,213 variants, respectively. The processing time for the HB method was 3.411 seconds while the ML method took 445.237 seconds to output the results. As we observe, the MM method, compared to the HB and ML methods, led to the lowest maximum of estimated LFDRs for those 2,213 SNPs. This may suggest that the MM method is more reliable in detecting the association of those SNPs than the ML and HB ones.}

 \textcolor{reviewer2}{It is remarkable to add that choosing a threshold for LFDRs in real data sets is a challenge. For example, if in the above CAD data set one chooses 0.01 as the LFDR threshold, the MM method identifies 4221 associated SNPs, which is almost twice as many as the number of associated SNPs when the threshold was 0.0003. Obviously, there is more confidence in identifying those 2213 associated SNPs rather than the new 4221 SNPs. In fact, the smaller is a chosen threshold, the higher is the confidence in identifying associated SNPs. The 0.01 LFDR threshold using the ML and HB methods leads to 3795 and 1716 associated SNPs, respectively. }



\subsection{Application to a microarray data set}
Although we presented our results for GWAS data, one may apply them in other contexts such as RNA gene expression studies. To illustrate this, we use a prostate data set used by \citet{Efron2012}, in which  genetic expression levels for 6033 genes were obtained for 102 men including 50 normal control individuals and 52 prostate cancer patients. The interest is to test whether there is any difference between gene expression levels and the prostate and normal individuals, $i=1,\ldots, 6033$. 

Let $\bar{y}_i^{(1)}$ and $\bar{y}_i^{(2)}$ be the mean of the normal individuals and cancer patients, and suppose that $s_i$ is an estimate of the pooled sample standard error. To conduct this multiple hypothesis testing problem, the two-sample $t$-test statistics $t_i=\frac{\bar{y}_i^{(1)}-\bar{y}_i^{(2)}}{s_i}$ need to be computed first.  One may then convert these test statistics to standard normal statistics $z_i=\Phi^{-1}(F_{100}(t_i))$, where $\Phi$ and $F_{100}$ are the cumulative distribution functions of normal distribution and $t$ distribution with 100 degrees of freedom, respectively. By this transformation, the null hypothesis can be expressed as $H_{0i}:z_i\sim N(0,1)$. Now, to apply our proposed method, it suffices to use the transformation $x_i=z_i^2$. Then, the multiple hypothesis testing problem reduces to testing $H_{0i}:x_i\sim \chi_1^2$.  \textcolor{reviewer2}{From equation (\ref{p0.hat}), the MM method with a processing time of 0.004 seconds on a personal computer (core i7, 3.5 GHz speed with 16 GB of RAM) led to $\widehat\pi_0=0.9364$ and $\widehat\lambda=4.5240$. The ML approach with $c=0$ and $d=10$ in equation (\ref{ML_arg}) and a processing time of 0.370 seconds led to $\widehat\pi_0=0.9443$ and $\widehat\lambda=4.9472$. The HB approach with a processing time of 0.069 seconds led to $\widehat\pi_0=0.9315$.}

 \textcolor{reviewer2}{Taking 0.01 as the LFDR threshold, the MM approach identified one gene to be differentially expressed, and the HB and ML approaches identified one and two differentially expressed genes, respectively. We also examined the impact of LFDR threshold change in the number of differentially expressed genes. For example, a 0.05 threshold led both the MM and ML approaches to identically identifying 13 differentially expressed genes, while it led the HB approach to identify 3 genes only.}


\section{Discussion and concluding remarks}\label{sec:7}
In this paper, we investigated estimating LFDRs for genetic association data. By reviewing well-known measures of association in the literature, we showed that many of the currently used measures reduce to a chi-square model with one degree of freedom. We presented a simple LFDR estimation strategy by using the MM estimators of the proportion $\pi_0$ and non-centrality parameter $\lambda$. The approach, as presented in Theorem \ref{thm2} as well as Section \ref{sec:6}, is simple and fast to apply. Also, as demonstrated  by the two simulation strategies in Section \ref{sec:5} and the real data analyses in Section \ref{sec:6}, it leads to reliable estimates. On the other hand, the ML approach of \citet{Marta2012} highly depends on the bounds of $c$ and $d$ in (\ref{ML_arg}), is time consuming, and the processing time increases with the number of SNPs as well as the length on the interval $[c,d]$. The HB approach of  \citet{Efron2012} also depends on some preset parameters such as the number of breaks in the discretization of the $z$-scores, the degrees of freedom for fitting the estimated density, etc., and it may fail due to model misfit.

 \textcolor{reviewer1}{As a limitation of our proposed method, it uses a parametric model. Another limitation is that the parametric model relies on the assumption that all non-null features have the same non-centrality parameter.  This might not seem biologically realistic, but there are important advantages behind. This assumption makes the estimation procedure easy and straightforward. Of course, having different non-centrality parameters in the model makes it more biologically realistic, but that would rise the issue of interpretability as well as estimation complexity. A single non-centrality parameter in the model is in fact  a measure of the detectability of associations \citep{Bukszar2009}. It can also be interpreted as the average deviation of the data distributions of SNPs associated  with a disease from the data distribution of those unassociated SNPs \citep{Yang2013}. On the other hand, as discussed in Section \ref{sec:5}, it seems that the accuracy of LFDR estimation in both the MM and ML methods relies on the accuracy of estimating $\pi_0$ rather than $\lambda$. Therefore, having the same non-centrality parameter in the proposed model should not significantly impact the final list of associated SNPs.}

It is remarkable that our proposed approach is similar to the classic hypothesis testing in the sense that the test statistic $x_i$ is compared to a threshold. However, the threshold in the MM approach is $h_{u}(\widehat{\pi}_{0},\widehat{\lambda})$, while in the classic hypothesis testing it is just a $100(1-\frac \alpha 2)\%$  (or in some cases $100(1-\alpha)\%$) quantile of the underlying distribution. This ideal property provides a more user-friendly estimator of LFDR than the other existing approaches.

It is worth adding that, estimating LFDRs in the literature is usually done by using some algorithms without knowing explicit form of estimators of the underlying parameters. For example, in the ML estimation used by \citet{Marta2012}, an algorithm is applied to find arguments that maximize the likelihood function numerically, without  providing any closed form of the resulting estimators of the parameter. Such algorithms may also require some unrealistic assumptions such as independency. On the contrary, our proposed approach offers explicit forms of estimators of the parameters $\pi_0$ and $\lambda$, without imposing any restriction to the model.  This leads to user-friendly estimators using the simple Bayes rule provided in equation (\ref{eq:NewBayes rule 0-1 loss-1}). All a user needs is the test statistics and estimated values of $\pi_0$ and $\lambda$.

As discussed in our second real data analysis, the estimated value of $\pi_0$ using our proposed approach is very close to the HB estimate of \citet{Efron2012}. Obviously, this compliance confirms that the two approaches estimate $\pi_0$ very well, but this does not mean that the same threshold of estimated LFDR should be used. This is due to the fact that  \citet{Efron2012}'s LFDR estimation is based on the normal model while our proposed approach is on the basis of the chi-square model. Thus, if one is interested in using a 0.20 threshold when using Efron's HB approach, she/he might use a different one (maybe 0.1 or less) when applying our proposed approach.

\textcolor{reviewer2}{The estimation procedure presented in this paper can be used for other purposes, too. For example, \citet{karimnezhad2020incorporating} introduce LFDR estimation in presence of some additional information such as genetic annotations. They use the ML approach of  \citet{Marta2012} but one may be interested in applying our proposed estimation approach in their reference class problem. If so, new estimators of LFDR will be derived.}



\begin{acknowledgement}
The author is grateful to two anonymous reviewers for their constructive comments. Prostate data  are available online through http://statweb.stanford.edu/\\~ckirby/brad/LSI/datasets-and-programs/datasets.html. The data on coronary artery disease have been contributed by CARDIoGRAMplusC4D investigators 
  and have been downloaded from www.CARDIOGRAMPLUSC4D.ORG. The HB and ML based LFDR estimates have been computed using the \textit{locfdr} \citep{Rlocfdr} and \textit{LFDR.MLE} \citep{LFDR.ML} packages, respectively.
\end{acknowledgement}
\vspace*{1pc}




\begin{thebibliography}{}

\bibitem[{Benjamini and Hochberg(1995)}]{Benjamini1995}
{Benjamini Y, Hochberg Y} (1995) Controlling the false discovery rate: A practical and powerful approach to multiple testing. J Royal Stat Soc, Series B 57:289-300

\bibitem[{Bickel(2013)}]{Bickel2013}
{Bickel DR} (2013) Simple estimators of false discovery rates given as few as one or two p-values without strong parametric assumptions. Stat Appl in Genet Mol Biol 12(4):529-543


\bibitem[{Buksz\'{a}r et. al(2009)}]{Bukszar2009}
Buksz\'{a}r J, McClay JL, van den Oord EJ (2009) Estimating the posterior probability that genome-wide association findings are true or false. Bioinform 25(14):1807-1813


\bibitem[{Clarke et~al.(2011)}]{Clarke2011}
{Clarke GM}, {Anderson CA}, {Pettersson FH}, {Cardon LR}, {Morris AP}, {Zondervan KT} (2011) Basic statistical analysis in genetic case-control studies. Nat Protoc 6(2):121-133


\bibitem[{Efron(2004)}]{Efron2004}
Efron B (2004) Large-scale simultaneous hypothesis testing: The choice of a null hypothesis. J Am Stat Assoc 99:96-104

\bibitem[{Efron(2007)}]{Efron2007}
{Efron B} (2007) Correlation and large-scale simultaneous significance testing. J Am Stat Assoc 102:93-103


\bibitem[{Efron(2012)}]{Efron2012}
{Efron B} (2012) Large-scale inference: empirical Bayes methods for estimation, testing, and prediction. Cambridge University Press, New York



\bibitem[{Efron et.~al(2001)}]{Efronetal2001}
{Efron B, Tibshirani R, Storey JD, Tusher V} (2001) Empirical Bayes analysis of a microarray experiment. J Am Stat Assoc 96(456):1151-1160


\bibitem[{Efron et~al.(2011)}]{Rlocfdr}
Efron B, Turnbull BB, Narasimhan B (2011) locfdr: Computes local false  discovery rates. Reference Manual,~R package version 1.1-7



\bibitem[Hochberg(1988)]{Hochberg1988}
{Hochberg Y} (1988) A sharper Bonferroni procedure for multiple tests of significance. Biom 75(4):800-802


\bibitem[Holm(1979)]{Holm1979}
{Holm S} (1979) A simple sequentially rejective multiple test procedure. Scand J Stat 6(2):65-70


\bibitem[{Karimnezhad and Bickel(2020)}]{karimnezhad2020incorporating}
{Karimnezhad A, Bickel DR} (2020) Incorporating prior knowledge about genetic variants into the analysis of genetic association data: An empirical Bayes approach.  IEEE/ACM Trans  Comput Biol and Bioinform 17(2):635-464


\bibitem[{Harris and M\'{a}ty\'{a}s(1999)}]{Matyas1999}
Harris D, M\'{a}ty\'{a}s L (1999) Introduction to the generalized method of moments estimation. In: M\'{a}ty\'{a}s L (ed) Generalized method of moments estimation. Cambridge University Press, New York, pp 3-30


\bibitem[{Muralidharan(2010)}]{Muralidharan2010}
{Muralidharan O} (2010) An empirical Bayes mixture method for effect size and false discovery rate estimation. Ann Appl Stat 4(1):422-438



\bibitem[{Nikpay et. al(2015)}]{Nikpay2015}Nikpay M, Goel A, Won HH, Hall LM, Willenborg C, Kanoni S, et al. (2015) A comprehensive 1000 Genomes- based genome-wide association meta-analysis of coronary artery disease. Nat Genet 47(10):1121-1130



\bibitem[{Padilla and Bickel(2012)}]{Marta2012}
{Padilla M, Bickel DR} (2012) Estimators of the local false discovery rate designed for small numbers of tests. Stat Appl Genet Mol Biol 11(5) Art. 4

\bibitem[{Pan et~al.(2003)}]{pan2003mixture}
{Pan W, Lin J, Le CT} (2003) A mixture model approach to detecting differentially expressed genes with microarray data. Funct Integr Genom 3(3):117-124


\bibitem[Slatkin(2008)]{Slatkin2008}
Slatkin M (2008) Linkage disequilibrium--understanding the evolutionary past and mapping the medical future. Nat Rev Genet 9(6):477-485


\bibitem[Shao(2007)]{Shao2007}
{Shao J} (2007) Mathematical statistics, 2nd ed. Springer-Verlag, New York


\bibitem[Sid\'{a}k(1968)]{Sidak1968}
Sid\'{a}k Z (1968) On multivariate normal probabilities of rectangles: their dependence on correlations. Ann Math Stat 39(5):1425-1434

\bibitem[Sid\'{a}k(1971)]{Sidak1971}
Sid\'{a}k Z. (1971) On probabilities of rectangles in multivariate Student distributions: their dependence on correlations.  Ann Math Stat 42(1):169-175

\bibitem[Simes(1986)]{Simes1986}
{Simes RJ} (1986) An improved Bonferroni procedure for multiple tests of significance. Biometrika 73(3):751-754


\bibitem[{Storey(2002)}]{Storey2002}
{Storey JD} (2002) A direct approach to false discovery rates. J Royal Stat Soc, Series B 64:479-498


\bibitem[{Yang et~al.(2013)}]{Yang2013}
Yang Y, Aghababazadeh FA, Bickel DR (2013) Parametric estimation of the local false discovery rate for identifying genetic associations. IEEE/ACM Trans  Comput Biol and Bioinform 10:98-108

\bibitem[{Yang et~al.(2015)}]{LFDR.ML}
Yang Y, Padilla M, Ali A, Leckett K, Yang Z, Li Z (2015) LFDR.MLE. Reference Manual,~R package version 1.1-10


\bibitem[{Zhao et.~al(2013)}]{Zhao2013}Zhao Z, Wang W, Wei Z (2013) An empirical Bayes testing procedure for detecting variants in analysis of next generation sequencing data. Ann Appl Stat 7(4):2229-2248


\bibitem[{Zheng et~al.(2012)}]{Zheng2012}
Zheng G, Yang Y, Zhu X, Elston RC (2012) Analysis of genetic association studies. Springer Science and Business Media, New York



\end{thebibliography}
\end{document}